\documentclass[twocolumn,english]{article}
\usepackage[T1]{fontenc}
\usepackage[latin9]{inputenc}
\usepackage{geometry}
\usepackage{color}
\usepackage{babel}
\usepackage{float}
\usepackage{amsthm}
\usepackage{amsmath}
\usepackage{amssymb}
\usepackage{graphicx}
\usepackage{esint}
\usepackage[numbers]{natbib}
\usepackage[unicode=true,pdfusetitle,
 bookmarks=true,bookmarksnumbered=false,bookmarksopen=false,
 breaklinks=false,pdfborder={0 0 0},backref=false,colorlinks=true]
 {hyperref}
\hypersetup{
 citecolor=blue,urlcolor=blue}

\makeatletter

\floatstyle{ruled}
\newfloat{algorithm}{tbp}{loa}
\providecommand{\algorithmname}{Algorithm}
\floatname{algorithm}{\protect\algorithmname}

  \theoremstyle{plain}
  \newtheorem*{assumption*}{\protect\assumptionname}
  \theoremstyle{plain}
  \newtheorem*{thm*}{\protect\theoremname}
  \theoremstyle{plain}
  \newtheorem{prop}{\protect\propositionname}
  \theoremstyle{definition}
  \newtheorem{defn}{\protect\definitionname}
  \theoremstyle{remark}
  \newtheorem*{rem*}{\protect\remarkname}

\usepackage{abstract}
\usepackage{tikz}
\usepackage{fancybox}
\usetikzlibrary{calc}
\usetikzlibrary{shapes,snakes}

\tikzstyle{trigl}=[isosceles triangle,  
   draw, 
   shape border rotate=90, 
   inner sep=3, 
   font=\small\sffamily\bfseries,
   isosceles triangle apex angle=60,
   isosceles triangle stretches]

\newdimen\XCoord
\newdimen\YCoord
\newdimen\Bottom

\@ifundefined{showcaptionsetup}{}{%
 \PassOptionsToPackage{caption=false}{subfig}}
\usepackage{subfig}
\makeatother

  \providecommand{\assumptionname}{Assumption}
  \providecommand{\definitionname}{Definition}
  \providecommand{\propositionname}{Proposition}
  \providecommand{\remarkname}{Remark}
  \providecommand{\theoremname}{Theorem}

\begin{document}
\newcommand{\iid}{\stackrel{\mathrm{iid}}{\sim}}

\title{Forest resampling for distributed sequential Monte Carlo}

\author{
Anthony Lee\thanks{Department of Statistics, University of Warwick} \and 
Nick Whiteley\thanks{School of Mathematics, University of Bristol}\footnotemark[2]}

\twocolumn[
\maketitle 

{\bf Abstract}:
This paper brings explicit considerations of distributed computing architectures and data structures into the rigorous design of Sequential Monte Carlo (SMC) methods. A theoretical result established recently by the authors shows that adapting interaction between particles to suitably control the Effective Sample Size (ESS) is sufficient to guarantee stability of SMC algorithms. Our objective is to leverage this result and devise algorithms which are thus guaranteed to work well in a distributed setting. We make three main contributions to achieve this. Firstly, we study mathematical properties of the ESS as a function of matrices and graphs that parameterize the interaction amongst particles. Secondly, we show how these graphs can be induced by tree data structures which model the logical network topology of an abstract distributed computing environment. Thirdly, we present efficient distributed algorithms that achieve the desired ESS control, perform resampling and operate on forests associated with these trees. \bigskip

 \bigskip \bigskip
]
\saythanks

\section{Introduction}

SMC algorithms are interacting particle methods for approximating
sequences of distributions arising in statistics, and are commonly
applied to Hidden Markov Models (HMM's) for filtering and marginal
likelihood estimation (see, e.g., \citep{doucet2000sequential,Doucet2008}).
We focus here on this HMM setting for simplicity, although our methodology
is relevant to other SMC schemes, such as \citep{smc:meth:C02}, \citep{DelMoral2006}
and \citep{chopin2013smc2}. It is becoming increasingly important
that computationally intensive algorithms are suited to implementation
on many-core computing architectures (see, e.g., \citep{Suchard2009}),
and it is well established that standard SMC algorithms naturally
have this property (see, e.g., \citep{lee2010utility}). In particular,
the time in which such algorithms run on many-core devices is typically
sublinear in the number of particles, $N$, until $N$ reaches a device-
and application-specific critical size, resulting in significant performance
improvements for moderate numbers of particles. However, the number
of particles required for acceptable accuracy in various settings
can be substantially larger than this critical size. In order to provide
accurate estimates in these situations in a timely fashion, attention
is naturally drawn to distributed implementations of SMC algorithms,
in which particles are distributed over multiple devices which can
communicate over a network (see, amongst others, \citep{bolic2005resampling,jun2012entangled,verge2013parallel}).
In this environment the interactions between particles, which provide
fundamental stability properties of the algorithm, are costly due
to relatively slow network speeds in comparison to fast on-device
memory accesses.

Motivated by the desire to develop Monte Carlo algorithms whose communication
structure is more naturally suited to distributed architectures, \citet{whiteley2013role}
proposed and studied a generalization of standard SMC algorithms,
called $\alpha$SMC, in which interaction between particles may be
modulated in an on-line fashion. The ``$\alpha$'' in $\alpha$SMC
refers to certain matrices which are chosen adaptively as the algorithm
runs, dictating or constraining this interaction. A special case of
$\alpha$SMC is the popular adaptive resampling strategy originally
proposed by \citet{liu1995blind}. One of the main results of \citep{whiteley2013role}
is a stability theorem which shows that, subject to regularity conditions
on the HMM, adapting $\alpha$ so as to enforce an appropriate lower
bound on the ESS is sufficient to ensure time-uniform convergence
of $\alpha$SMC filtering estimates, and endow it with other attractive
theoretical properties so that the computational cost of the algorithm
grows manageably with the length of the data record. This provides
a criterion for stabilization of these algorithms when communication
constraints influence interaction.

Monitoring and controlling the ESS using $\alpha$ matrices is therefore
very important. However, if implemented naively, this monitoring and
control itself involves collective operations on the entire particle
system, and so remains as an obstacle to parallelization. In this
paper, our overall aim is to address this obstacle and formulate approaches
to ESS control which are more appropriate for distributed implementation.
In order to do so, we consider a logical tree topology which represents
an abstract distributed computing environment. This network structure
accommodates divide-and-conquer routines and recursive programming,
making it suited to distributed computation, and its hierarchical
nature lends itself to partitioning and resampling operations. We
consider methods of ESS control involving computations which are \emph{local}
with respect to the topology of these trees. 

After outlining $\alpha$SMC in Section~\ref{sec:asmc}, our first
original contribution in Section~\ref{sec:effective_sample_size}
is a study of the ESS itself, as a functional of the $\alpha$ matrix
governing interaction. This study leads us to consider a subset of
potential $\alpha$ matrices with a specific associated graphical
structure. We then define a partial order on this set of matrices,
which makes precise a sense in which they are more or less suited
to distributed architectures, and prove that the ESS is (partial)
order-preserving. This important relationship connects computational
considerations with statistical performance and informs our algorithm
design. Section~\ref{sec:effective_sample_size} culminates in a
lower bound on the ESS phrased in terms of particle sub-populations,
and applied recursively this bound leads to an abstract recursive
algorithm for enforcing a lower bound on the population-wide ESS.
Crucially, each recursive call of this algorithm can require the consideration
of only a small number of aggregated weights, and this is what makes
it suited to distributed architectures. Section~\ref{sec:Forest-resampling}
is devoted to practical implementation of this abstract recursive
algorithm in a distributed setting using trees, in such a way that
all quantities required are available via local computations whose
cost is independent of $N$. An interpretation of the resulting resampling
scheme is that it corresponds to a tree sampling procedure involving
a number of disjoint trees, and so we term the overall procedure forest
resampling. All proofs are given in the appendix.

\section{$\alpha$SMC\label{sec:asmc}}

In this section we overview relevant aspects of the general methodology
proposed in \citep{whiteley2013role}. An HMM with measurable state
space $\left(\mathsf{X},\mathcal{X}\right)$ and observation space
$\left(\mathsf{Y},\mathcal{Y}\right)$ is a process $\left\{ \left(X_{n},Y_{n}\right);n\geq0\right\} $
where $\left\{ X_{n};n\geq0\right\} $ is a Markov chain on $\mathsf{X}$,
the observations $\left\{ Y_{n};n\geq0\right\} $, valued in $\mathsf{Y}$,
are conditionally independent given $\left\{ X_{n};n\geq0\right\} $,
and the conditional distribution of each $Y_{n}$ depends on $\left\{ X_{n};n\geq0\right\} $
only through $X_{n}$. Let $\pi_{0}$ and $f$ be respectively a probability
distribution and a Markov kernel on $\left(\mathsf{X},\mathcal{X}\right)$,
and let $g$ be a Markov kernel acting from $\left(\mathsf{X},\mathcal{X}\right)$
to $\left(\mathsf{Y},\mathcal{Y}\right)$, with $g(x,\cdot)$ admitting
a density, denoted similarly by $g(x,y)$, with respect to some dominating
$\sigma$-finite measure. The HMM specified by $\pi_{0}$, $f$ and
$g$, is
\begin{eqnarray}
 &  & X_{0}\sim\pi_{0},\nonumber \\
 &  & X_{n}\mid\left\{ X_{n-1}=x_{n-1}\right\} \sim f(x_{n-1},\cdot),\quad n\geq1,\nonumber \\
 &  & Y_{n}\mid\left\{ X_{n}=x_{n}\right\} \sim g(x_{n},\cdot),\quad n\geq0.\label{eq:HMM}
\end{eqnarray}
Throughout this paper we consider a fixed observation sequence $\left\{ y_{n};n\geq0\right\} $
and write 
\begin{equation}
g_{n}(x):=g(x,y_{n}),\quad n\geq0.\label{eq:g_n}
\end{equation}
Throughout this paper we shall work under the mild assumption that
for each $n\geq0$, $\sup_{x\in\mathsf{X}}g_{n}(x)<+\infty$ and $g_{n}(x)>0$
for all $x\in\mathsf{X}$. 

For $n\geq1$, let $\pi_{n}$ be the conditional distribution of $X_{n}$
given $Y_{0:n-1}=y_{0:n-1}$, called the \emph{prediction filter};
and let $Z_{n}$ be the marginal likelihood of the first $n$ observations,
evaluated at the point $y_{0:n-1}$. Due to the conditional independence
structure of the HMM the following recursions hold:
\[
\pi_{n}\left(A\right)=\frac{\int_{\mathsf{X}}\pi_{n-1}\left(dx\right)g_{n-1}(x)f(x,A)}{\int_{\mathsf{X}}\pi_{n-1}\left(dx\right)g_{n-1}(x)},\quad A\in\mathcal{X},\; n\geq1,
\]
and
\[
Z_{n}=Z_{n-1}\int_{\mathsf{X}}\pi_{n-1}\left(dx\right)g_{n-1}\left(x\right),\quad n\geq1,
\]
with the convention $Z_{0}:=1$. Our main computational objectives
are to approximate $\left\{ \pi_{n};n\geq0\right\} $ and $\left\{ Z_{n};n\geq0\right\} $.

We write $[M]:=\{1,\ldots,M\}$ for a generic $M\in\mathbb{N}$. We
denote by $N$ an arbitrary but fixed positive integer representing
the number of particles in the algorithm we are about to describe.\emph{
}To simplify presentation, whenever a summation sign appears without
the summation set made explicit, the summation set is taken to be
$[N]$, for example we write $\Sigma_{i}$ to mean $\Sigma_{i=1}^{N}$.

Let $\mathbb{A}_{[N]}$ be the set of doubly stochastic matrices of
size $N\times N$ (this is a special case of the setup of \citep{whiteley2013role},
corresponding to their assumption $({\bf B}^{++})$). The $\alpha$SMC
algorithm simulates a sequence $\left\{ \zeta_{n};n\geq0\right\} $
with each $\zeta_{n}:=\left(\zeta_{n}^{1},\ldots,\zeta_{n}^{N}\right)$
valued in $\mathsf{X}^{N}$. When $n\geq1$, this involves choosing
a matrix $\alpha_{n-1}$ from $\mathbb{A}_{[N]}$ according to some
deterministic function of $\left\{ \zeta_{0},\ldots,\zeta_{n-1}\right\} $,
and this matrix specifies the type of interaction that occurs at time
$n$. 

\begin{algorithm}[H]
\begin{raggedright}
For $n=0$,
\par\end{raggedright}

\begin{raggedright}
\qquad{}For $i=1,\ldots,N$,
\par\end{raggedright}

\begin{raggedright}
\qquad{}\qquad{}Set\quad{} $W_{0}^{i}=1$.
\par\end{raggedright}

\begin{raggedright}
\qquad{}\qquad{}Sample\quad{} $\zeta_{0}^{i}\sim\pi_{0}$.
\par\end{raggedright}

\begin{raggedright}
For $n\geq1$,
\par\end{raggedright}

\begin{raggedright}
$(\star)$\hspace*{0.25cm}Select $\alpha_{n-1}$ from $\mathbb{A}_{[N]}$
as a function of $\left\{ \zeta_{0},\ldots,\zeta_{n-1}\right\} $
\par\end{raggedright}

\begin{raggedright}
\qquad{}For $i=1,\ldots,N$,
\par\end{raggedright}

\begin{raggedright}
$(\dagger)$\qquad{}Set\quad{} $W_{n}^{i}=\sum_{j}\alpha_{n-1}^{ij}W_{n-1}^{j}g_{n-1}(\zeta_{n-1}^{j})$.
\par\end{raggedright}

\begin{raggedright}
$(\ddagger)$\qquad{}Sample\quad{}
\begin{eqnarray*}
 &  & \zeta_{n}^{i}|\zeta_{0},\ldots,\zeta_{n-1}\\
 &  & \quad\sim\dfrac{\sum_{j}\alpha_{n-1}^{ij}W_{n-1}^{j}g_{n-1}(\zeta_{n-1}^{j})f(\zeta_{n-1}^{j},\cdot)}{\sum_{j}\alpha_{n-1}^{ij}W_{n-1}^{j}g_{n-1}(\zeta_{n-1}^{j})}
\end{eqnarray*}

\par\end{raggedright}

\protect\caption{$\alpha$SMC\label{alg:aSMC}}
\end{algorithm}

With $\delta_{x}$ denoting the Dirac measure centred on $x$, the
objects
\begin{equation}
\pi_{n}^{N}:=\frac{\sum_{i}W_{n}^{i}\;\delta_{\zeta_{n}^{i}}}{\sum_{i}W_{n}^{i}},\quad Z_{n}^{N}:=\frac{1}{N}\sum_{i}W_{n}^{i},\quad n\geq0,\label{eq:particle_approximations}
\end{equation}
are regarded as approximations of $\pi_{n}$ and $Z_{n}$, respectively. 

In general, some algorithm design is involved at line $(\star)$ of
Algorithm~\ref{alg:aSMC}; one has to decide on a rule which dictates
how $\alpha_{n-1}$ is chosen from $\mathbb{A}_{[N]}$, and in practice
one will often select $\alpha_{n-1}$ from $\mathbb{A}_{[N]}$ as
some function of $\left(W_{n-1}^{1},\ldots,W_{n-1}^{N}\right)$ and
$\left(g_{n-1}(\zeta_{n-1}^{1}),\ldots,,g_{n-1}(\zeta_{n-1}^{N})\right)$.
Two members of $\mathbb{A}_{[N]}$ used implicitly in methods predating
$\alpha$SMC are: $\mathbf{1}_{1/N}$, the $N\times N$ matrix which
has $1/N$ as every entry; and $Id$, the identity matrix. If $\alpha_{n}=\mathbf{1}_{1/N}$
for every $n$, $\alpha$SMC reduces to the bootstrap particle filter,
whereas if $\alpha_{n}=Id$ for every $n$, $\alpha$SMC reduces to
sequential importance sampling. If at each time step one chooses adaptively
between $\mathbf{1}_{1/N}$ and $Id$, $\alpha$SMC is equivalent
to the adaptive resampling method of \citep{liu1995blind}. We refer
the reader to \citep[Section 2.2]{whiteley2013role} for the details
of these equivalences.

The ESS associated with the weights $\left\{ W_{n}^{i}:i\in[N]\right\} $is
\begin{equation}
N_{n}^{\text{eff}}:=\frac{\left(\sum_{i}W_{n}^{i}\right)^{2}}{\sum_{i}\left(W_{n}^{i}\right)^{2}}.\label{eq:ESS_defn_front}
\end{equation}
Looking also at line $(\dagger)$ of Algorithm~\ref{alg:aSMC}, we
see that $N_{n}^{\text{eff}}$ clearly depends on $\alpha_{n-1}$
but not on $\zeta_{n}$. Therefore $\alpha_{n-1}$ can be selected
adaptively to ensure that $N_{n}^{\text{eff}}$ exceeds some threshold
\emph{before $\zeta_{n}$} is simulated. In this paper we investigate
methods to carry out this kind of adaptive selection, with $\alpha_{n-1}$
chosen from a large family of matrices which includes $\mathbf{1}_{1/N}$
and $Id$.

One of the main contributions of \citep{whiteley2013role} is the
stability theorem stated below, which gives a rigorous theoretical
justification for enforcing a lower bound on $N_{n}^{\text{eff}}$.
This theorem relies on the following regularity condition on the HMM,
which is often used to establish stability results for non-adaptive
SMC algorithms (see e.g., \citep{smc:the:DMG01,smc:the:CdMG11,whiteley2014twisted},
and see also \citep{whiteley2013stability} for stability under weaker
conditions).
\begin{assumption*}
$\mathbf{(C)}$ There exists $\left(\delta,\epsilon\right)\in[1,\infty)^{2}$
such that
\[
\sup_{n\geq0}\sup_{x,y}\frac{g_{n}(x)}{g_{n}(y)}\leq\delta,\quad\quad f(x,\cdot)\leq\epsilon f(y,\cdot),\quad(x,y)\in\mathsf{X}^{2}.
\]

\end{assumption*}
For $\mu$ a measure on $(\mathsf{X},\mathcal{X})$ and $\varphi$
a real-valued, $\mathcal{X}$-measurable function on $\mathsf{X}$
we define $\mu(\varphi):=\int_{\mathsf{X}}\varphi(x)\mu(dx)$, allowing
us to compare $\pi_{n}^{N}$ with $\pi_{n}$ via the differences $\pi_{n}^{N}(\varphi)-\pi_{n}(\varphi)$,
for suitable $\varphi$. For example, when $A\in\mathcal{X}$ and
$\varphi=\mathbf{1}_{A}$ then $\pi_{n}(\varphi)$ is the conditional
probability that $X_{n}\in A$ given $Y_{0:n-1}=y_{0:n-1}$ and $\pi_{n}^{N}(\varphi)$
its $\alpha$SMC estimate.
\begin{thm*}
\label{thm:L_R_mix}\citep[ Theorem 2]{whiteley2013role} Assume $\mathbf{(C)}$.
Then\textup{\emph{ there exist finite constants $c_{1}$ and for any
$r\geq1$, $c_{2}(r)$, such that for any $N\geq1$ and $\tau\in(0,1]$,
if 
\begin{equation}
\inf_{n\geq0}N_{n}^{\text{eff}}\geq N\tau,\label{eq:ess_lower_bounded_condition}
\end{equation}
then
\[
\sup_{n\geq1}\;\mathbb{E}\left[\left(\dfrac{Z_{n}^{N}}{Z_{n}}\right)^{2}\right]^{1/n}\leq1+\dfrac{c_{1}}{N\tau},
\]
and for any $\varphi:\mathsf{X}\rightarrow\mathbb{R}$ which is $\mathcal{X}$-measurable
and bounded,
\[
\sup_{n\geq0}\;\mathbb{E}\left[\left|\pi_{n}^{N}(\varphi)-\pi_{n}(\varphi)\right|^{r}\right]^{1/r}\leq\left\Vert \varphi\right\Vert _{\infty}\dfrac{c_{2}(r)}{\sqrt{N\tau}}.
\]
}}
\end{thm*}
In this paper our objective is to design instances of $\alpha$SMC
which guarantee (\ref{eq:ess_lower_bounded_condition}) whilst achieving
a desirable balance between the communication costs associated with
steps $(\star)$ and $(\ddagger)$ of Algorithm~\ref{alg:aSMC}.
\citet[Section 5.3]{whiteley2013role} suggested some procedures for
adaptively selecting $\alpha_{n-1}$ from $\mathbb{A}_{[N]}$ at line
$(\star)$. However, a practical issue concerning these adaptive procedures
is that guaranteeing (\ref{eq:ess_lower_bounded_condition}) involves
evaluating the ESS for some candidate $\alpha_{n-1}$'s, and this
task may itself be demanding in terms of communication cost. Indeed,
if one wishes to search through a large set of candidates for $\alpha_{n-1}$,
e.g. when attempting to guarantee (\ref{eq:ess_lower_bounded_condition})
with as sparse an $\alpha_{n-1}$ as possible, the cost of step $(\star)$
may dominate the overall cost of Algorithm~\ref{alg:aSMC}. 

On the other hand, the adaptive resampling particle filter \citep{liu1995blind}
involves only the two candidates $Id$ and $\mathbf{1}_{1/N}$; evaluating
the ESS for the candidate $Id$ can be done cheaply, and if $\alpha_{n-1}=\mathbf{1}_{1/N}$,
then we always have $N_{n}^{\text{eff}}=N$, so step $(\star)$ is
inexpensive. However, if $\alpha_{n-1}=Id$ does not achieve $N_{n}^{\text{eff}}\geq N\tau$
there is no choice but to set $\alpha_{n-1}=\mathbf{1}_{1/N}$, and
one then incurs the communication cost associated with the resulting
population-wide interaction at step $(\ddagger)$. 

To help us understand how we can achieve (\ref{eq:ess_lower_bounded_condition})
using sparse $\alpha_{n-1}$, but without excessive communication,
we proceed with an investigation of the ESS.

\section{Properties of the ESS\label{sec:effective_sample_size}}

\subsection{Dependence of the ESS on $\alpha$\label{sub:Properties-of-the_ESS}}

Slightly extending our notation, for each non-empty $V\subseteq[N]$
let $\mathbb{A}_{V}$ be the set of all substochastic $N\times N$
matrices $a$ with the following properties:
\begin{enumerate}
\item $a$ leaves the uniform distribution on $V$ invariant,
\item $a^{ij}=0$ whenever $(i,j)\in[N]^{2}\setminus V^{2}$. 
\end{enumerate}
Note that when $V=[N]$, we have $\mathbb{A}_{V}\equiv\mathbb{A}_{[N]}$
as defined in Section~\ref{sec:asmc}. By convention, when $V=\emptyset$,
we define $\mathbb{A}_{V}$ to contain only the zero matrix. It is
readily observed that if $a\in\mathbb{A}_{V}$ and $a'\in\mathbb{A}_{V'}$
with $V\cap V'=\emptyset$ then $(a+a')\in\mathbb{A}_{V\cup V'}$.

Now let $\mathbb{A}:=\bigcup_{V\subseteq[N]}\mathbb{A}_{V}$ and define
the function $N^{{\rm eff}}:\mathbb{A}\times\mathbb{R}_{+}^{N}\rightarrow\mathbb{R}_{+}$,
\begin{equation}
N^{{\rm eff}}(a,c):=\begin{cases}
0 & \; a\in\mathbb{A}_{\emptyset},\\
\frac{\left(\sum_{i}\sum_{j}a^{ij}c^{j}\right)^{2}}{\sum_{i}\left(\sum_{j}a^{ij}c^{j}\right)^{2}} & \;\text{otherwise},
\end{cases}\label{eq:ESS}
\end{equation}
where $c=(c^{1},\ldots,c^{N})\in\mathbb{R}_{+}^{N}$ (for simplicity
we shall always assume that each $c^{i}$ is strictly positive). This
generalizes the ESS in (\ref{eq:ESS_defn_front}): let $c$ be given
by $c^{i}:=W_{n-1}^{i}g_{n-1}(\zeta_{n-1}^{i}),\: i\in[N]$. Then,
if $\alpha_{n-1}=a\in\mathbb{A}_{[N]}$, we have $N^{{\rm eff}}(\alpha_{n-1},c)\equiv N_{n}^{{\rm eff}}$.
If instead $\alpha_{n-1}=a+a'$, where $a\in\mathbb{A}_{V}$ for some
strict, non-empty subset $V\subset[N]$ and $a'\in\mathbb{A}_{[N]\setminus V}$,
then 
\[
N^{{\rm eff}}(a,c)=\frac{\left(\sum_{i\in V}\sum_{j\in V}a^{ij}c^{j}\right)^{2}}{\sum_{i\in V}\left(\sum_{j\in V}a^{ij}c^{j}\right)^{2}}=\frac{\left(\sum_{i\in V}W_{n}^{i}\right)^{2}}{\sum_{i\in V}\left(W_{n}^{i}\right)^{2}}
\]
represents the ESS associated with the sub-population of $\left|V\right|$
weights $\left\{ W_{n}^{i}:i\in V\right\} $, cf. (\ref{eq:ESS_defn_front}). 

The following proposition provides useful properties of $N^{{\rm eff}}$.
\begin{prop}
\label{prop:ess_properties}Let $V,V'\subseteq[N]$ such that $V\cap V'=\emptyset$.
Let $a\in\mathbb{A}_{V}$, $a',\tilde{a}'\in\mathbb{A}_{V'}$ and
$c\in\mathbb{R}_{+}^{N}$ be given such that $N^{{\rm eff}}\left(a,c\right)$,
$N^{{\rm eff}}\left(a',c\right)$ and $N^{{\rm eff}}\left(\tilde{a}',c\right)$
are all positive. All of the following hold:\end{prop}
\begin{enumerate}
\item Extremes: $1\leq N^{{\rm eff}}(a,c)\leq|V|$ and $N^{{\rm eff}}(a,c)=|V|$
whenever $a^{ij}=|V|^{-1}\mathbb{I}\left(i,j\in V\right)$.
\item Subadditivity: 
\[
N^{{\rm eff}}\left(a+a',c\right)\leq N^{{\rm eff}}\left(a,c\right)+N^{{\rm eff}}\left(a',c\right),
\]
with equality only when $\frac{\sum_{j\in V}c^{j}}{N^{{\rm eff}}(a,c)}=\frac{\sum_{j\in V'}c^{j}}{N^{{\rm eff}}(a',c)}$.
\item Monotonicity:
\begin{multline*}
N^{{\rm eff}}(a',c)\leq N^{{\rm eff}}(\tilde{a}',c)\\
\Longrightarrow N^{{\rm eff}}(a+a',c)\leq N^{{\rm eff}}(a+\tilde{a}',c),
\end{multline*}
with equality on the right hand side of the implication only when
$N^{{\rm eff}}(a',c)=N^{{\rm eff}}(\tilde{a}',c)$. 
\item Lower bound: 
\[
N^{{\rm eff}}\left(a+a',c\right)\geq\min\left\{ N^{{\rm eff}}\left(a,c\right),N^{{\rm eff}}\left(a',c\right)\right\} .
\]

\end{enumerate}
The first part of this proposition is well known and identifies extremal
values of $N^{{\rm eff}}$, of which the maximal value can always
be realized by a particular choice of $a$. The other parts concern
properties of $N^{{\rm eff}}$ when considering elements of $\mathbb{A}_{V}$
and $\mathbb{A}_{V'}$ for some fixed and disjoint $V,V'\subseteq[N]$.
The second establishes the subadditivity of $N^{{\rm eff}}$ and indicates
that the effective sample size associated with $a+a'$ is less than
the sum of those associated with $a$ and $a'$ separately. The third
shows that nevertheless a monotonicity property holds when comparing
two substochastic matrices in $\mathbb{A}_{V}$, and the fourth provides
a simply-proved but tight lower bound on the effective sample size
associated with $a+a'$.

\subsection{Disjoint unions of complete graphs and a partial order\label{sub:Graphs,-induced-substochastic}}

\citet[ Section 5.3]{whiteley2013role} considered a family of candidate
$\alpha$ matrices which have the interpretation of being transition
matrices of random walks on regular undirected graphs. In this section,
we expand upon this duality between $\alpha$ matrices and undirected
graphs, and introduce some mathematical machinery which allows us
describe how these objects are related to each other, and $N^{{\rm eff}}$.
In particular, we consider graphs that are disjoint unions of complete
graphs (Definition~\ref{def:(Disjoint-union-of} below): these graphs
are not necessarily regular but are highly structured nonetheless
and are of interest here because we can define a partial order over
them, and then establish a partial order preservation result for $N^{{\rm eff}}$
(see Propositions~\ref{prop:partial_order} and~\ref{prop:order_preserving})
that will ultimately guide the efficient exploration of progressively
denser stochastic matrices until one is found, $a$, for which we
can guarantee $N^{{\rm eff}}(a,c)\geq N\tau$.

To proceed, let us introduce some standard graph-theoretic notions.
A graph $G=(V,E)$ is a set of vertices $V\subseteq[N]$ and a set
of edges $E\subseteq V^{2}$, where an edge $(i,j)\in E$ represents
a connection between vertices $i$ and $j$. We adopt the convention
that $(i,i)\in E$ whenever $i\in V$. If $G$ is undirected then
$(i,j)\in E\iff(j,i)\in E$. If $G$ is a complete graph then $E=V^{2}$.
Since a complete graph is defined solely by its vertex set, and because
complete graphs are important building blocks in the sequel, we define
$\kappa(V):=(V,V^{2})$ to be the complete graph with vertices $V$.

Let $\cup$ denote the disjoint union of two graphs: if $G=(V,E)$,
$G'=(V',E')$ and $V\cap V'=\emptyset$ then $G\cup G'=(V\cup V',E\cup E')$.

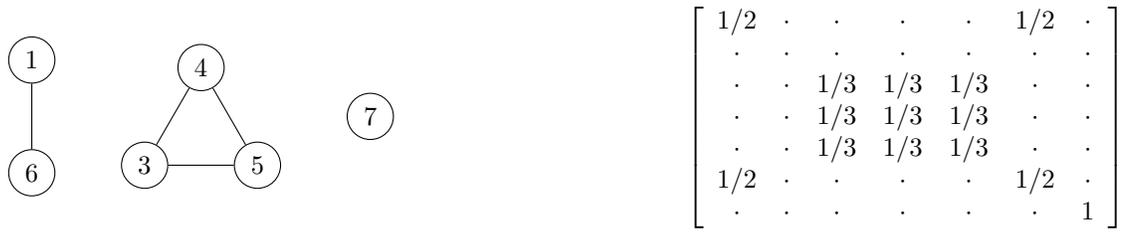
\begin{figure*}
\center
\begin{tabular}{cc}
\mbox{
\begin{minipage}{.55\textwidth}
\begin{center}
\begin{tikzpicture}
\def \s {1.5}
\node[draw, circle] at (\s*1,\s*.5)(v2) {$1$};
\node[draw, circle] at (\s*1,\s*-.5)(v3) {$6$};
\node[draw, circle] at (\s*2,\s*-.5*0.8660254)(v4) {$3$};
\node[draw, circle] at (\s*2.5,\s*.5*0.8660254)(v5) {$4$};
\node[draw, circle] at (\s*3,\s*-.5*0.8660254)(v6) {$5$};
\node[draw, circle] at (\s*4,\s*0)(v7) {$7$};
\path[every node/.style={font=\sffamily\small}]
    (v2) edge node [left] {} (v3)
    (v4) edge node [left] {} (v5)
    (v5) edge node [left] {} (v6)
    (v4) edge node [left] {} (v6)
    ;
\end{tikzpicture}
 \end{center}
\end{minipage}}
&  \mbox{
\begin{minipage}{.4\textwidth}
\begin{equation*}
\left[\begin{array}{ccccccc}
1/2&\cdot&\cdot&\cdot&\cdot&1/2&\cdot\\
\cdot&\cdot&\cdot&\cdot&\cdot&\cdot&\cdot\\ \cdot&\cdot&1/3&{1}/{3}&1/3&\cdot&\cdot\\
\cdot&\cdot&1/3&1/3&1/3&\cdot&\cdot\\
\cdot&\cdot&1/3&1/3&1/3&\cdot&\cdot\\
1/2&\cdot&\cdot&\cdot&\cdot&1/2&\cdot\\
\cdot&\cdot&\cdot&\cdot&\cdot&\cdot&1\\
\end{array}\right]
\end{equation*}
\end{minipage}}
\end{tabular}
\caption{A graph $G\in\mathbb{G}$, with vertex set $\{1,3,4,5,6,7\}$, and the corresponding matrix $\phi(G)$. For visual clarity, self-loops are not shown.}
\label{fig:complete_graphs_mtx}
\end{figure*}

\begin{figure*}
\center
\begin{tabular}{c@{}c@{}c}
\mbox{
\begin{minipage}{.5\textwidth}
\begin{center}
\begin{tikzpicture}
\def \s {1.5}
\node[draw, circle] at (\s*1,\s*.5)(v2) {$1$};
\node[draw, circle] at (\s*1,\s*-.5)(v3) {$6$};
\node[draw, circle] at (\s*2,\s*-.5*0.8660254)(v4) {$3$};
\node[draw, circle] at (\s*2.5,\s*.5*0.8660254)(v5) {$4$};
\node[draw, circle] at (\s*3,\s*-.5*0.8660254)(v6) {$5$};
\node[draw, circle] at (\s*4,\s*0)(v7) {$7$};
\path[every node/.style={font=\sffamily\small}]
    (v2) edge node [left] {} (v3)
    (v4) edge node [left] {} (v5)
    (v5) edge node [left] {} (v6)
    (v4) edge node [left] {} (v6)
    ;
\end{tikzpicture}
\end{center}
\end{minipage}}
& {\huge$\preceq$} &
\mbox{
\begin{minipage}{.38\textwidth}
\begin{center}
\begin{tikzpicture}
\def \s {1.5}
\node[draw, circle] at (\s*1,\s*.5)(v2) {$1$};
\node[draw, circle] at (\s*1,\s*-.5)(v3) {$6$};
\node[draw, circle] at (\s*2,\s*-.5)(v4) {$3$};
\node[draw, circle] at (\s*2,\s*.5)(v5) {$4$};
\node[draw, circle] at (\s*3,\s*-.5)(v6) {$5$};
\node[draw, circle] at (\s*3,\s*.5)(v7) {$7$};
\path[every node/.style={font=\sffamily\small}]
    (v2) edge node [left] {} (v3)
    (v4) edge node [left] {} (v5)
    (v5) edge node [left] {} (v6)
    (v4) edge node [left] {} (v6)
    (v7) edge node [left] {} (v4)
    (v7) edge node [left] {} (v5)
    (v7) edge node [left] {} (v6)
    ;
\end{tikzpicture}
 \end{center}
\end{minipage}}
\end{tabular}
\caption{Two graphs $G,G^\prime\in\mathbb{G}$ with $G\preceq G^\prime$.}
\label{fig:complete_graphs_order}
\end{figure*}
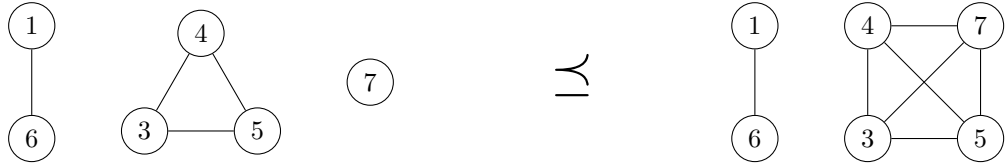
\begin{defn}
\label{def:(Disjoint-union-of} (Disjoint union of complete graphs)
A graph $G$ is a disjoint union of complete graphs if for some $K\in[N]$
there exists a set of pairwise disjoint subsets of $[N]$, denoted$\{V_{k}:k\in[K]\}$
such that $G=\bigcup_{k\in[K]}\kappa(V_{k})$.
\end{defn}
In analogy with $\mathbb{A}_{V}$ (and $\mathbb{A}$) we define $\mathbb{G}_{V}$
to be the set of graphs which have vertices $V$ and which are disjoint
unions of complete graphs (and $\mathbb{G}:=\bigcup_{V\subseteq[N]}\mathbb{G}_{V}$).
Clearly, if $G\in\mathbb{G}_{V}$, $G'\in\mathbb{G}_{V'}$ and $V\cap V'=\emptyset$,
then $G\cup G'\in\mathbb{G}_{V\cup V'}$. We also define the matrix-valued
function $\phi$ 

\[
\phi:G=(V,E)\in\mathbb{G}\longmapsto a=(a^{ij})\in\mathbb{A}
\]
where
\begin{equation}
a^{ij}:=\begin{cases}
\frac{\mathbb{I}\left\{ (i,j)\in E\right\} }{\sum_{k\in[N]}\mathbb{I}\left\{ (i,k)\in E\right\} } & \;(i,j)\in V^{2},\\
0 & \;(i,j)\in[N]^{2}\setminus V^{2}.
\end{cases}\label{eq:phi_def}
\end{equation}
One trivial property of elements $G=(V,E)\in\mathbb{G}$ is that $(i,j)\in E$
and $(j,k)\in E$ implies $(i,k)\in E$. It is therefore clear that
if $G\in\mathbb{G}_{V}$ then $\phi(G)$ is a symmetric matrix and
leaves the uniform distribution on $V$ invariant, hence, $\phi(G)\in\mathbb{A}_{V}$.
Figure~\ref{fig:complete_graphs_mtx} shows an example of a graph
$G\in\mathbb{G}$ and the corresponding substochastic matrix $\phi(G)$.
Letting $\mathbb{A}_{\mathbb{G}}:=\{\phi(G):G\in\mathbb{G}\}$ be
the image of $\phi$, it is straightforward that $\phi:\mathbb{G}\rightarrow\mathbb{A}_{\mathbb{G}}$
is a bijection and so we denote by $\phi^{-1}$ the inverse of $\phi$.
In addition, it can be seen that if $G\in\mathbb{G}_{V}$ and $G'\in\mathbb{G}_{V'}$
with $V\cap V'=\emptyset$, then 
\begin{equation}
\phi\left(G\cup G'\right)=\phi(G)+\phi(G').\label{eq:phi_union_addition}
\end{equation}
 We can now introduce a particular relation amongst graphs, and amongst
the corresponding substochastic matrices.
\begin{defn}
\label{def:binary_reln} (Binary relation $\preceq$) Let $G=(V,E)$
and $\tilde{G}=(\tilde{V},\tilde{E})$ be members of $\mathbb{G}$.
Then we write $G\preceq\tilde{G}$ if and only if $V=\tilde{V}$ and
$E\subseteq\tilde{E}$. Since $\phi$ is a bijection between $\mathbb{G}$
and $\mathbb{A}_{\mathbb{G}}$ we will also write, for $a,\tilde{a}\in\mathbb{A}_{\mathbb{G}}$,
$a\preceq\tilde{a}$ if and only $\phi^{-1}(a)\preceq\phi^{-1}(\tilde{a})$.\end{defn}
\begin{prop}
\label{prop:partial_order} (Partial order) $\preceq$ is a partial
order over $\mathbb{G}$ and $\mathbb{A}_{\mathbb{G}}$. 
\end{prop}
Definition~\ref{def:binary_reln} says that for some $G,\tilde{G}\in\mathbb{G}_{V}$
we have $G\preceq\tilde{G}$ if $\tilde{G}=G$, or if $\tilde{G}$
can be obtained from $G$ by adding edges in such a way that $\tilde{G}\in\mathbb{G}_{V}$.
Intuitively, one can imagine adding edges by choosing two of the complete
graphs comprising $G$ and adding edges between all vertices in these
two graphs. Figure~\ref{fig:complete_graphs_order} shows an example
of two graphs $G,G^{\prime}\in\mathbb{G}$ such that $G\preceq G'$.
We note that $\preceq$ is not a total order, because there exist
members of $\mathbb{G}_{V}$, $G=(V,E)$, $G'=(V,E'$) such that $E\nsubseteq E'$
and $E'\nsubseteq E$. Our interest in $\preceq$ is the following
order preservation property. 
\begin{prop}
\label{prop:order_preserving} (Order preservation) For any $c\in\mathbb{R}_{+}^{N}$,
$a\preceq\tilde{a}\implies N^{{\rm eff}}(a,c)\leq N^{{\rm eff}}(\tilde{a},c)$.
\end{prop}

\subsection{Local lower bounds on $N^{{\rm eff}}$\label{sub:Local-lower-bounds}}

In this subsection we present Algorithm~\ref{alg:Choosea}, a recursive
method for efficient selection of $a\in\mathbb{A}_{[N]}$; using a
corresponding recursive lower bound on $N^{{\rm eff}}$ (Proposition~\ref{prop:alpha_lower_bound})
and the ordering result Proposition~\ref{prop:order_preserving},
we shall validate Algorithm~\ref{alg:Choosea} with Proposition~\ref{prop:recursiveboundinduction},
which shows that it is guaranteed to achieve $N^{{\rm eff}}(a,c)\geq N\tau$. 

For purposes of exposition, we first provide an expression for $N^{{\rm eff}}(\tilde{a},c)$
when $\tilde{a}$ is the substochastic matrix associated with a disjoint
union of complete graphs. Following (\ref{eq:phi_union_addition}),
let $\tilde{a}=\sum_{k\in[K]}\phi(\kappa(V_{k}))\in\mathbb{A}_{\mathbb{G}}\cap\mathbb{A}_{V}$
for some $K\in[N]$ and pairwise disjoint $\{V_{k}:k\in[K]\}$ with
$V=\bigcup_{k\in[K]}V_{k}$. Then, from (\ref{eq:ESS}),
\begin{eqnarray}
N^{{\rm eff}}\left(\tilde{a},c\right) & = & \frac{\left(\sum_{k\in[K]}\sum_{j\in V_{k}}c^{j}\right)^{2}}{\sum_{k\in[K]}\sum_{i\in V_{k}}\left(\sum_{j\in V_{k}}\phi(\kappa(V_{k}))^{ij}c^{j}\right)^{2}}\nonumber \\
 & = & \frac{\left(\sum_{k\in[K]}\sum_{j\in V_{k}}c^{j}\right)^{2}}{\sum_{k\in[K]}|V_{k}|\left(|V_{k}|^{-1}\sum_{j\in V_{k}}c^{j}\right)^{2}}.\label{eq:ess_union_complete}
\end{eqnarray}
We note that (\ref{eq:ess_union_complete}) depends on $c$ only through
the values of the sums $\left\{ \sum_{j\in V_{k}}c^{j}:k\in[K]\right\} $;
we can interpret this as saying that $N^{{\rm eff}}\left(\tilde{a},c\right)$
is equal to the ESS associated with a collection of $\sum_{k\in[K]}|V_{k}|$
weights, in which for each $k\in[K]$ there are $|V_{k}|$ weights
all taking the value $|V_{k}|^{-1}\sum_{j\in V_{k}}c^{j}$. This lends
interpretation to the lower bound in the following proposition. 
\begin{prop}
\label{prop:alpha_lower_bound}Let $\{V_{k}:k\in[K]\}$ consist of
non-empty and pairwise disjoint subsets of $[N]$ and $\{a_{k}:k\in[K]\}$
be given such that each $a_{k}\in\mathbb{A}_{V_{k}}$. Let $a=\sum_{k\in[K]}a_{k}$
and $\tilde{a}=\sum_{k\in[K]}\phi(\kappa(V_{k}))$. Then for any $c\in\mathbb{R}_{+}^{N}$,
\begin{equation}
N^{{\rm eff}}\left(\tilde{a},c\right)\geq N^{{\rm eff}}\left(a,c\right)\geq\min_{k}\left\{ \frac{N^{{\rm eff}}(a_{k},c)}{|V_{k}|}\right\} N^{{\rm eff}}\left(\tilde{a},c\right).\label{eq:N^eff_lower_bound_recursion}
\end{equation}

\end{prop}
Importantly, Proposition~\ref{prop:alpha_lower_bound} enables us
to calculate a lower bound on $N^{{\rm eff}}\left(\sum_{k\in[K]}a_{k},c\right)$
without explicit computation of (\ref{eq:ESS}). This observation
is at the heart of our new algorithms.

A disjoint union of complete graphs with vertices $V\subseteq[N]$
can be succinctly represented by a partition $P=\{V_{k}:k\in[K]\}$
of $V$, where $K\in[\left|V\right|]$. Overloading our $N^{{\rm eff}}(\cdot,c)$
notation so as to conveniently express certain quantities in Algorithm~\ref{alg:Choosea},
we define for such a partition $P$,
\begin{equation}
N^{{\rm eff}}(P,c):=\frac{\left(\sum_{S\in P}\sum_{j\in S}c^{j}\right)^{2}}{\sum_{S\in P}|S|\left(|S|^{-1}\sum_{j\in S}c^{j}\right)^{2}}.\label{eq:Neff_partition_formulation}
\end{equation}
Since $P$ is a partition of $V$, we have 
\begin{equation}
\frac{N^{{\rm eff}}(P,c)}{\sum_{S\in P}\sum_{j\in S}1}=|V|^{-1}N^{{\rm eff}}(P,c)=:\rho(P,c),\label{eq:fp_first}
\end{equation}
and this quantity also appears in Algorithm~\ref{alg:Choosea}.

If $P$ and $\tilde{P}$ are the partitions representing $G$ and
$\tilde{G}$ respectively, where $G,\tilde{G}\in\mathbb{G}_{V}$ for
some $V\subseteq[N]$, then $G\preceq\tilde{G}$ if and only if $P$
is a refinement of $\tilde{P}$. This allows us to make the following
definition, which will be used extensively in the sequel.
\begin{defn}
\label{def:coarsening} (Coarsening) Let $P$, $\tilde{P}$ be partitions
of some subset of $[N]$. Then $\tilde{P}$ is a coarsening of $P$,
written $\tilde{P}\succeq P$, if and only if $P$ is a refinement
of $\tilde{P}$.
\end{defn}
It follows from Proposition~\ref{prop:order_preserving} that $\tilde{P}\succeq P\implies N^{{\rm eff}}(\tilde{P},c)\geq N^{{\rm eff}}(P,c)$.

\begin{algorithm}
\protect\caption{Choose an $a\in\mathbb{A}_{V}$ such that $N^{{\rm eff}}\left(\sum_{k\in[K]}a_{k},c\right)\geq\tau|V|$\label{alg:Choosea}}

\texttt{$\mathtt{choose.a}(V,\tau)$}
\begin{enumerate}
\item Choose a partition $P$ of $V$ such that $\rho(P,c)\geq\tau$.
\item If $P=\{V\}$ then return $\phi(\kappa(V))$.
\item Otherwise, return 
\[
\sum_{k\in[K]}{\tt choose.a}\left(V_{k},\tau/\rho(P,c)\right).
\]
\end{enumerate}
\end{algorithm}

\begin{prop}
\label{prop:recursiveboundinduction}Algorithm~\ref{alg:Choosea}
called with $(V,\tau)$ satisfying $\emptyset\neq V\subseteq[N]$
and $\tau\in[0,1]$ returns $a\in\mathbb{A}_{V}$ such that $N^{{\rm eff}}\left(a,c\right)\geq\tau|V|$.
\end{prop}
There are a number of ways that step 1 of Algorithm~\ref{alg:Choosea}
can be implemented. One possibility, motivated by Proposition~\ref{prop:order_preserving},
is to search through a sequence of successively coarser, candidate
partitions until the condition $\rho(P,c)\geq\tau$ is met. In Section~\ref{sec:Forest-resampling}
we provide a more detailed and practical version of this procedure
in Algorithm~\ref{alg:forest_constrained}, in which the partitions
considered arise from collections of tree data structures.

\section{Forest resampling\label{sec:Forest-resampling}}

In this section we introduce tree data structures to represent the
logical topology of a distributed computer architecture. Loosely,
these trees provide a model for how the operations involved in $\alpha$SMC
can be arranged over a network of communicating devices, each of which
has the capacity to store data and to perform basic simulation and
arithmetic tasks. In Sections~\ref{sub:distributed_arch}--\ref{sub:Trees-from-architecture}
we explain the connection between the distributed architecture and
tree data structures, and in Section~\ref{sub:graphs_induced_trees_forests}
we explain the connection between trees and forests, and the partitions,
graphs and matrices addressed in Section~\ref{sec:effective_sample_size}.
Sections~\ref{sec:Forest-resampling} and~\ref{sub:forest_selection}
describe the role of forests when implementing respectively lines
$(\ddagger)$, $(\dagger)$ and $(\star)$ of Algorithm~\ref{alg:aSMC},
and all these ingredients are brought together in Algorithm~\ref{alg:aSMCimpl},
which is an implementation of Algorithm~\ref{alg:aSMC} using trees
and forests.

\subsection{Distributed computer architecture\label{sub:distributed_arch}}

For the purposes of this paper, we are interested primarily in a setting
where there are a number of possibly heterogeneous computing devices
that can communicate via sending data over a network. Qualitatively,
the structural assumption will be that communication within a device
is far quicker than communication between devices. If there are $M$
devices, we might think each device $i\in[M]$ is capable of handling
a particle system with $N_{i}$ particles. This implies that interactions
involving the $N_{i}$ particles on device $i$ are considerably less
costly than interactions involving particles on different devices.

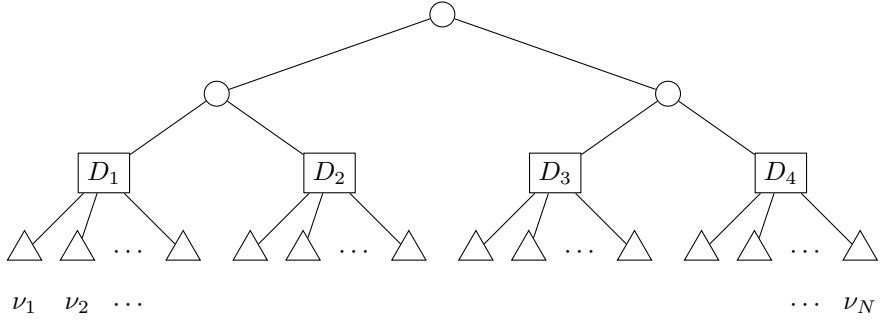
\begin{figure*}
\center
\begin{tikzpicture}[
	level/.style={sibling distance=60mm/#1},
	level distance=30pt,
	level 3/.style={sibling distance = 20pt},
	omit/.style={edge from parent/.style={red,very thick,draw=none}}
]
\def \boffset {20pt}
 \node[draw, circle] (v1) {$$}
 	child {
		node [circle,draw] (node1) {$$}
 		child {
			node [rectangle,draw] (D1) {$D_1$}
 			child {node [trigl] (n1) {$$} }
	 			child {node [trigl] (n2) {$$} }
	 			child[omit] {node [draw=none] (n3) {$\cdots$} }
	 			child {node [trigl] (n4) {$$} }
	 		}
	 		child {
			node [rectangle,draw] (D2) {$D_2$}
 			child {node [trigl] (n5) {$$} }
	 			child {node [trigl] (n6) {$$} }
	 			child[omit] {node [draw=none] (n7) {$\cdots$} }
	 			child {node [trigl] (n8) {$$}
 			}
	 		}
	}
 	child {
		node [circle,draw] (node2) {$$}
 		child {
			node [rectangle,draw] (D3) {$D_3$}
 			child {node [trigl] (n9) {$$} }
	 			child {node [trigl] (n10) {$$} }
	 			child[omit] {node [draw=none] (n11) {$\cdots$} }
	 			child {node [trigl] (n12) {$$} }	 		}
	 		child {
			node [rectangle,draw] (D4) {$D_4$}
 			child {node [trigl] (n13) {$$} }
	 			child {node [trigl] (n14) {$$} }
	 			child[omit] {node [draw=none] (n15) {$\cdots$} }
	 			child {node [trigl] (n16) {$$}
 			}
	 		}
	}
 ;
\path (n1);
\pgfgetlastxy{\XCoord}{\Bottom};
\node[draw=none, circle] (T1) at (\XCoord,\Bottom-\boffset) {$\nu_1$};
\path (n2);
\pgfgetlastxy{\XCoord}{\YCoord};
\node[draw=none, circle] (T2) at (\XCoord,\Bottom-\boffset) {$\nu_2$};
\path (n3); \pgfgetlastxy{\XCoord}{\YCoord};
\node[draw=none, circle] (T3) at (\XCoord,\Bottom-\boffset) {$\cdots$};
\path (n16); \pgfgetlastxy{\XCoord}{\YCoord}; \node[draw=none, circle] (T4) at (\XCoord,\Bottom-\boffset) {$\nu_N$}; \path (n15); \pgfgetlastxy{\XCoord}{\YCoord};
\node[draw=none, circle] (T5) at (\XCoord,\Bottom-\boffset) {$\cdots$};
\end{tikzpicture}
\caption{Roles of nodes.}
\label{fig:node_roles}
\end{figure*}

\subsection{Trees from architecture\label{sub:Trees-from-architecture}}

The architecture described in Section~\ref{sub:distributed_arch}
suggests the use of a particular type of data structure, a tree, to
represent possible interactions between computing devices. A tree
is a recursive data structure comprising a set of nodes with associated
values.
\begin{defn}
(Node) A node $\nu$ is an object that has a value, $\mathcal{V}(\nu)$,
and a (possibly empty) set of child nodes, $\mathcal{C}(\nu)$.
\end{defn}

\begin{defn}
(Finite tree) A (finite) tree $T$ is a finite set of nodes which
is either empty, or satisfies the following properties:
\begin{enumerate}
\item $\mathcal{C}(\nu)\subseteq T$ for every $\nu\in T$ (no node has
children outside $T$).
\item $\mathcal{C}(\nu)\cap\mathcal{C}(\nu')=\emptyset$ for any distinct
$\nu,\nu'\in T$ (no node is the child of two different nodes in $T$).
\item There exists a unique element of $T$ called the root and denoted
$\mathcal{R}(T)$, such that $\mathcal{R}(T)\notin\bigcup_{\nu\in T}\mathcal{C}(\nu)$
(a unique root node is not a child of any of the other nodes in $T$).
\end{enumerate}
\end{defn}
One can show (e.g., by contradiction) that if $T$ is a tree then
every node in $T$ other than $\mathcal{R}(T)$ is a descendant of
$\mathcal{R}(T)$, i.e., $T\setminus\{\mathcal{R}(T)\}=\mathcal{D}(\mathcal{R}(T))$
where $\mathcal{D}(\nu)$ denotes the descendants of $\nu$:
\[
\mathcal{D}(\nu):=\begin{cases}
\emptyset & \mathcal{C}(\nu)=\emptyset,\\
\mathcal{C}(\nu)\cup\left(\bigcup_{\chi\in\mathcal{C}(\nu)}\mathcal{D}(\chi)\right) & \mathcal{C}(\nu)\neq\emptyset.
\end{cases}
\]

\begin{defn}
\label{def:subtree} (Subtree) A subtree, of a tree $T$, consists
of a node in $T$, taken together with all of the descendants of that
node. In particular, for some $\nu\in T$ we call $\mathcal{S}(\nu):=\nu\cup\mathcal{D}(\nu)$
the subtree of $T$ with root $\nu$. 
\end{defn}
The definitions of a tree and its subtrees are equivalent to those
found in \citep[p. 308]{knuth2005art}, but with an emphasis on their
formulation using children. Here, trees serve as data structures in
that the value of each node is the data stored there, and data transfer
can occur between a node and its children.

It is conventional to call a node of a tree $T$ whose set of children
is empty a leaf, and the set of such nodes comprise the leaves of
$T$. Our intention is to have the individual particles, indexed by
$j\in[N]$, represented by leaves of a tree and the parents of leaves
representing the $M$ devices in the distributed architecture. If
each device $i$ is assigned $N_{i}$ particles then the children
of the node associated with device $i$ will be the $N_{i}$ leaves
associated with the particle indices $\left\{ 1+\sum_{j\in[i-1]}N_{j},\ldots,\sum_{j\in[i]}N_{j}\right\} $.
Beyond these two levels, the structure is purposefully abstract so
as to accommodate various choices which could, e.g., be related to
more complex architectural considerations such as the geographical
location of the devices. It is, however, assumed that each node is
physically contained on a single device although more than one node
may be physically contained on the same device. The general idea is
that a node will both facilitate and modulate interaction between
its children. Figure~\ref{fig:node_roles} shows a possible tree
with $4$ devices.

Let $T_{0}$ be a tree with root node $\nu_{0}$ and exactly $N$
leaves $\{\nu_{i}:i\in[N]\}$. We now define the set of leaf indices
associated with a node $\nu$ of $T_{0}$ to be the set of indices
associated with the leaves of $\mathcal{S}(\nu)$, i.e., we let $\ell(\nu_{i}):=\{i\}$
for each $i\in[N]$, and for each $\nu\in T_{0}$ such that $\mathcal{C}(\nu)\neq\emptyset$,
$\ell(\nu):=\bigcup_{\chi\in\mathcal{C}(\nu)}\ell(\chi)$. Without
ambiguity we also define, for $T$ a subtree of $T_{0}$, $\ell(T):=\ell(\mathcal{R}(T))$.
For some $c\in\mathbb{R}_{+}^{N}$, we define the value of each node
$\nu$ to be
\[
\mathcal{V}(\nu):=\left(\mathcal{V}_{1}(\nu),\mathcal{V}_{2}(\nu)\right):=\left(\left|\ell(\nu)\right|,\sum_{j\in\ell(\nu)}c^{j}\right),
\]
so that the value of leaf node $\nu_{i}$, e.g., is $\mathcal{V}(\nu_{i})=(1,c^{i})$.
Once the values of the leaves have been set, Algorithm~\ref{alg:populate}
can be invoked on $\nu_{0}$ to calculate recursively the values of
the rest of the nodes in the tree, and is motivated by the fact that,
element-wise,
\begin{equation}
\mathcal{V}(\nu)=\sum_{\chi\in\mathcal{C}(\nu)}\mathcal{V}(\chi),\label{eq:value_nodes_children}
\end{equation}
when $\mathcal{C}(\nu)\neq\emptyset$. This is an instance of a recursive
reduction algorithm suitable for implementation in both parallel and
distributed settings (see, e.g., \citep{hillis1986data,isard2007dryad})
which can be called on the root of the subtree in question. Typically,
one will call it on $\nu_{0}$ to populate the entire tree $T_{0}$.
The time complexity associated with each node $\nu$'s computation
is in $\mathcal{O}(|\mathcal{C}(\nu)|)$.

\begin{algorithm}
\protect\caption{Populate a subtree\label{alg:populate}}

\texttt{$\mathtt{populate}(\nu)$}
\begin{enumerate}
\item If $\mathcal{C}(\nu)=\emptyset$, return $\mathcal{V}(\nu)$.
\item Otherwise, set $\mathcal{V}(\nu)\leftarrow\sum_{\chi\in\mathcal{C}(\nu)}\mathtt{populate}(\chi)$,
where the summation is component-wise.
\item Return $\mathcal{V}(\nu)$.\end{enumerate}
\end{algorithm}

\subsection{Graphs induced by trees and forests\label{sub:graphs_induced_trees_forests}}

We now take the first step towards connecting our tree data structures
with the type of graphs discussed in Section~\ref{sec:effective_sample_size}.
We define the graph induced by a tree $T$ to be 
\begin{equation}
G(T):=\kappa(\ell(T)),\label{eq:graph_tree}
\end{equation}
the complete graph with vertices $\ell(T)$. This allows us to define
the substochastic matrix induced by a tree as $\phi(T):=\phi(G(T))$.
It is immediately obvious that the only member of $\mathbb{A}_{[N]}$
that can be induced by a single tree is $\phi(T_{0})={\bf 1}_{1/N}$.
The notion of a forest allows a richer subset of $\mathbb{A}_{[N]}$
to be specified using trees.
\begin{defn}
(Forest) A forest $F$ is a set of pairwise disjoint trees.
\end{defn}
It follows from this definition that if $T,T'\in F$ are distinct,
then $\ell(T)\cap\ell(T')=\emptyset$. If $T$ is a tree then $\{T\}$
and $\{\mathcal{S}(\nu):\nu\in\mathcal{C}(\mathcal{R}(T))\}$ are
both examples of forests. In what follows, the forests defined will
always be comprised of subtrees of $T_{0}$. Figure~\ref{fig:tree_forest_graphs}
supplements the example from Figure~\ref{fig:complete_graphs_mtx}
with a possible associated tree data structure and forest of subtrees.

\begin{figure*}
\center
\begin{tikzpicture}[level/.style={sibling distance=60mm/#1},level distance=20pt]
\def \boffset {20pt}
 \node[draw, circle] (v1) {$$}
	child {node [circle,draw,fill=black] (n1) {$$} }
 	child { 	 	node [circle,draw] (v2) {$$}
		child {
			node [circle,draw,fill=black] (v5) {$$}
			child {node [circle,draw,fill=black] (n2) {$$}}
			child {node [circle,draw,fill=black] (n3) {$$}}
		}
		child {
			node [circle,draw,fill=black] (v3) {$$}
			child {node [circle,draw,fill=black] (n4) {$$}}
			child {
				node [circle,draw,fill=black] (v4) {$$}
				child {node [circle,draw,fill=black] (n5) {$$}}
				child {node [circle,draw,fill=black] (n6) {$$}}
			}
		}
	}
 ;
\path (n6); \pgfgetlastxy{\XCoord}{\Bottom};
\path (n1); \pgfgetlastxy{\XCoord}{\YCoord};
\node[draw, circle] (T1) at (\XCoord,\Bottom-\boffset) {$7$};
\path (n2); \pgfgetlastxy{\XCoord}{\YCoord};
\node[draw, circle] (T2) at (\XCoord,\Bottom-\boffset) {$1$};
\path (n3); \pgfgetlastxy{\XCoord}{\YCoord};
\node[draw, circle] (T3) at (\XCoord,\Bottom-\boffset) {$6$};
\path (n4); \pgfgetlastxy{\XCoord}{\YCoord};
\node[draw, circle] (T4) at (\XCoord,\Bottom-\boffset) {$3$};
\path (n5); \pgfgetlastxy{\XCoord}{\YCoord};
\node[draw, circle] (T5) at (\XCoord,\Bottom-\boffset) {$4$};
\path (n6); \pgfgetlastxy{\XCoord}{\YCoord};
\node[draw, circle] (T6) at (\XCoord,\Bottom-\boffset) {$5$};
\path[every node/.style={font=\sffamily\small}]
    (T2) edge node [left] {} (T3)
    (T4) edge [bend right=40] node [left] {} (T6)
    (T4) edge node [left] {} (T5)
    (T5) edge node [left] {} (T6)
         
	(n1) edge [dashed] node [left] {} (T1)
	(n2) edge [dashed] node [left] {} (T2)
	(n3) edge [dashed] node [left] {} (T3)
	(n4) edge [dashed] node [left] {} (T4)
	(n5) edge [dashed] node [left] {} (T5)
	(n6) edge [dashed] node [left] {} (T6)

	(v5) edge [line width=2pt] node [left] {} (n2)
    (v5) edge [line width=2pt] node [left] {} (n3)
    (v3) edge [line width=2pt] node [left] {} (v4)
    (v3) edge [line width=2pt] node [left] {} (n4)
    (v4) edge [line width=2pt] node [left] {} (n5)
    (v4) edge [line width=2pt] node [left] {} (n6)
    ;
\end{tikzpicture}
\caption{A forest made up of subtrees of a tree, and the complete graphs induced by each tree in the forest.}
\label{fig:tree_forest_graphs}
\end{figure*}
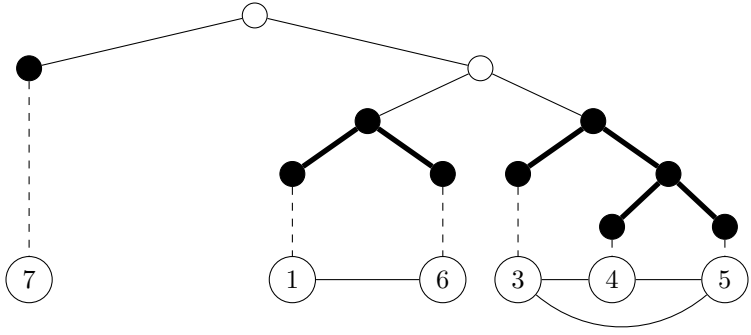

We define the set of leaf indices associated with a forest to be $\ell(F):=\bigcup_{T\in F}\ell(T)$.
We also let $\mathbb{F}_{V}:=\{F:\ell(F)=V\}$, where $V\subseteq[N]$,
and $\mathbb{F}:=\bigcup_{V\subseteq[N]}\mathbb{F}_{V}$. We can relate
any $F\in\mathbb{F}$ to a member of $\mathbb{G}$ by defining
\[
G(F):=\bigcup_{T\in F}G(T)=\bigcup_{T\in F}\kappa(\ell(T)).
\]
From (\ref{eq:phi_union_addition}), the substochastic matrix induced
by $F\in\mathbb{F}$ is then
\[
\phi(F):=\phi(G(F))=\sum_{T\in F}\phi(\kappa(\ell(T))).
\]
One can therefore think of a forest $F$ as being a data structure
counterpart to a disjoint union of complete graphs represented by
the partition $P=\left\{ \ell(T):T\in F\right\} $.

\subsection{Forest resampling\label{sub:forest_resampling}}

We now introduce practical methodology that, given a forest $F\in\mathbb{F}_{[N]}$,
enables implementation of step $(\ddagger)$ of Algorithm~\ref{alg:aSMC}
when $\alpha_{n-1}=\phi(F)$. Let $c$ be given by $c^{i}:=W_{n-1}^{i}g_{n-1}(\zeta_{n-1}^{i}),\: i\in[N]$,
so that our goal is to sample, for each $i\in[N]$,
\[
\zeta_{n}^{i}\mid\zeta_{0},\ldots,\zeta_{n-1}\;\sim\;\dfrac{\sum_{j}\alpha_{n-1}^{ij}c^{j}f(\zeta_{n-1}^{j},\cdot)}{\sum_{k}\alpha_{n-1}^{ik}c^{k}},
\]
which can be implemented in two substeps. First one simulates an ancestor
index $A_{n-1}^{i}$ with 
\[
\mathsf{P}\left(A_{n-1}^{i}=j\mid\zeta_{0},\ldots,\zeta_{n-1}\right)=\dfrac{\alpha_{n-1}^{ij}c^{j}}{\sum_{k}\alpha_{n-1}^{ik}c^{k}},
\]
and then, secondly, simulates $\zeta_{n}^{i}\sim f(\zeta_{n-1}^{A_{n-1}^{i}},\cdot)$.
Implementation of the second step is a model-specific matter, so we
focus on the first step. We define $t_{F}$ to be the tree-valued
map where for any $i\in\ell(F)$, $t_{F}(i)$ is the unique tree $T\in F$
such that $i\in\ell(T)$. It then follows that $\alpha_{n-1}^{ij}=\mathbb{I}\left\{ j\in\ell(t_{F}(i))\right\} /\left|\ell(t_{F}(i))\right|$
and so we can write
\begin{equation}
\mathsf{P}\left(A_{n-1}^{i}=j\mid\zeta_{0},\ldots,\zeta_{n-1}\right)=\dfrac{\mathbb{I}\left\{ j\in\ell(t_{F}(i))\right\} c^{j}}{\sum_{k\in\ell(t_{F}(i))}c^{k}},\label{eq:PrAij_forest}
\end{equation}
which implies that $A_{n-1}^{i}$ is categorically distributed over
$\ell(t_{F}(i))$ with probabilities proportional to $\left\{ c^{k}:k\in\ell(t_{F}(i))\right\} $.

Following (\ref{eq:PrAij_forest}), we propose Algorithm~\ref{alg:sample_tree},
which given the root node $\mathcal{R}(T)$ of an arbitrary subtree
$T$ of $T_{0}$, samples from a distribution over $\ell(T)$ with
probability mass function
\[
p_{T}(j):=\frac{\mathbb{I}\left\{ j\in\ell(T)\right\} c^{j}}{\sum_{k\in\ell(T)}c^{k}},\quad j\in\ell(T).
\]
Each recursive call of Algorithm~\ref{alg:sample_tree} with argument
$\nu$ has a time complexity in $\mathcal{O}\left(\left|\mathcal{C}(\nu)\right|\right)$.

\begin{algorithm}
\protect\caption{Obtain a sample according to $p_{\mathcal{S}(\nu)}$\label{alg:sample_tree}}

\texttt{$\mathtt{sample}(\nu)$}
\begin{enumerate}
\item If $\mathcal{C}(\nu)=\emptyset$, return the only element in $\ell(\nu)$.
\item Otherwise, let $\chi_{1},\ldots,\chi_{|\mathcal{C}(\nu)|}$ be the
children of $\nu$. 
\item Sample $i$ from a categorical distribution over $\left[\left|\mathcal{C}(\nu)\right|\right]$
with probabilities proportional to $\left\{ \mathcal{V}_{2}(\chi_{i}):i\in\left[\left|\mathcal{C}(\nu)\right|\right]\right\} $.
\item Return \texttt{$\mathtt{sample}(\chi_{i})$}.\end{enumerate}
\end{algorithm}

\begin{prop}
\label{prop:tree_sample_correctness}The probability that Algorithm~\ref{alg:sample_tree}
returns $j\in\ell(\nu)$ is $p_{\mathcal{S}(\nu)}(j)$.
\end{prop}
Sampling according to (\ref{eq:PrAij_forest}) for each $i\in[N]$
can be accomplished by calling Algorithm~\ref{alg:sample_tree} $N$
times with potentially different inputs. For example, if $F=\{T_{0}\}$
then one would call Algorithm~\ref{alg:sample_tree} $N$ times on
$\nu_{0}=\mathcal{R}(T_{0})$, corresponding to standard multinomial
resampling with $\alpha_{n-1}=\phi(F)=\mathbf{1}_{1/N}$. In contrast,
if $F=\left\{ \mathcal{S}(\nu_{1}),\ldots,\mathcal{S}(\nu_{N})\right\} $
then one would call Algorithm~\ref{alg:sample_tree} once on each
member of $\left\{ \nu_{1},\ldots,\nu_{N}\right\} $ with the effect
that $A_{n-1}^{i}=i$ for each $i\in[N]$, and this corresponds to
$\alpha_{n-1}=\phi(F)=Id$. An intermediate between these two extremes
would be if $F=\left\{ \mathcal{S}(\nu^{1}),\ldots,\mathcal{S}(\nu^{M})\right\} $,
where $\nu^{i}$ represents device node $i$ in $T_{0}$, cf. Section~\ref{sub:distributed_arch}.
Then, for each $i\in[M]$, one would call Algorithm~\ref{alg:sample_tree}
$\left|\ell(\nu^{i})\right|$ times, once to set each ancestor index
in $\left\{ A_{n-1}^{j}:j\in\ell(\nu^{i})\right\} $. These special
cases also exemplify a more general phenomenon: sampling according
to (\ref{eq:PrAij_forest}) using Algorithm~\ref{alg:sample_tree}
does not require the explicit computation of $\alpha_{n-1}$. In Section~\ref{sub:forest_selection}
we address the issue of how a forest can be chosen adaptively.

Finally, we note that step $(\dagger)$ of Algorithm~\ref{alg:aSMC}
can also be accomplished straightforwardly when $\alpha_{n-1}=\phi(F)$.
Indeed, then 
\begin{eqnarray*}
W_{n}^{i} & = & \sum_{j}\phi(F)^{ij}c^{j}=\sum_{k\in\ell(t_{F}(i))}c^{k}/\left|\ell(t_{F}(i))\right|\\
 & = & \mathcal{V}_{2}(\mathcal{R}(t_{F}(i)))/\mathcal{V}_{1}(\mathcal{R}(t_{F}(i))).
\end{eqnarray*}

\subsection{Forest selection\label{sub:forest_selection}}

Our attention now turns to implementing the $(\star)$ step of Algorithm~\ref{alg:aSMC}.
This can be performed by choosing a forest $F\in\mathbb{F}_{[N]}$
such that $N^{{\rm eff}}(\phi(F),c)\geq\tau N$. Algorithm~\ref{alg:forest_constrained}
is a recursive implementation of such a procedure, and is essentially
a practical analogue of Algorithm~\ref{alg:Choosea}. The $(\star\star)$
step in this algorithm is specified only abstractly, with concrete
choices the subject of Section~\ref{sub:partitioning_strategies}.
Like steps $(\dagger)$ and $(\ddagger)$ when implemented according
to the procedures of Section~\ref{sec:Forest-resampling}, step $(\star\star)$
also involves only local computations in the following sense. Recalling
Definition~\ref{def:coarsening}, choosing $P'$ to be a partition
of $\mathcal{C}(\nu)$ implies that $P$ is a coarsening of $\left\{ \ell(\chi):\chi\in\mathcal{C}(\nu)\right\} $,
and so the computation of $\rho(P,c)$ involves only the quantities
$\left|\ell(\chi)\right|$ and $\sum_{j\in\ell(\chi)}c^{j}$ for each
$\chi\in\mathcal{C}(\nu)$, which are readily available through $\left\{ \mathcal{V}(\chi):\chi\in\mathcal{C}(\nu)\right\} $.

\begin{algorithm}
\protect\caption{Specify a forest $F$ with $\ell(F)=\ell(\nu)$ and $N^{{\rm eff}}(\phi(F),c)\geq\tau|\ell(V)|$\label{alg:forest_constrained}}

\texttt{choose.forest}$(\nu,\tau)$
\begin{enumerate}
\item If $\mathcal{C}(\nu)=\emptyset$ then return $\{\mathcal{S}(\nu)\}$.
\item $(\star\star)$ Choose a partition $P'$ of $\mathcal{C}(\nu)$ such
that $\rho(P,c)\geq\tau$, where $P=\left\{ \bigcup_{\chi\in S}\ell(\chi):S\in P'\right\} $.
\item If $P'=\{\mathcal{C}(\nu)\}$ then return $\{\nu\}$. Otherwise, set
$R\leftarrow\emptyset$.
\item For each element $S\in P'$

\begin{enumerate}
\item If $|S|>1$ then create a node $\nu'$ with children $\{\chi:\chi\in S\}$
and set $R\leftarrow R\cup\{\mathcal{S}(\nu')\}$.
\item If $S=\{\chi\}$, set $R\leftarrow R\cup\mathtt{choose.forest}\left(\chi,\tau/\rho(P,c)\right)$.
\end{enumerate}
\item Return $R$.\end{enumerate}
\end{algorithm}

In Algorithm~\ref{alg:forest_constrained}, new nodes can be created.
It is assumed that when this happens, the values of the new nodes
are set appropriately according to (\ref{eq:value_nodes_children}). 

Gathering together Algorithms~\ref{alg:populate},~\ref{alg:sample_tree}
and~\ref{alg:forest_constrained} we now arrive at Algorithm~\ref{alg:aSMCimpl},
which is an implementation of Algorithm~\ref{alg:aSMC} using trees
and forests.

The recursive nature of the algorithms presented allow them to be
fairly straightforwardly translated into architecture specific implementations.
In particular, it is imagined that the computations of Algorithms~\ref{alg:populate},~\ref{alg:sample_tree}
and~\ref{alg:forest_constrained} all take place on the device on
which their node argument physically resides, and that the recursive
calls then represent messages passed over the network. In addition,
Algorithms~\ref{alg:populate} and~\ref{alg:forest_constrained}
are divide-and-conquer algorithms naturally suited to parallel implementation. 

The exact implementation of the algorithms may vary slightly, depending
on the architectures involved, without changing in principle. For
example, one implementation of Step 2e of Algorithm~\ref{alg:aSMCimpl}
could involve each device sending its list of associated indices ``up''
the tree until it reaches its root in the forest. From there, the
indices may filter ``down'' the tree in a slight variant of Algorithm~\ref{alg:sample_tree}
until they reach their leaves. If index $i$ reaches leaf $\nu_{j}$,
say, the device housing $\nu_{j}$ can send $\zeta_{n-1}^{j}$ to
the device housing $\nu_{i}$, which can then sample $\zeta_{n}^{i}\sim f(\zeta_{n-1}^{j},\cdot)$.

\begin{algorithm}
\protect\caption{$\alpha$SMC with forest resampling\label{alg:aSMCimpl}}

\begin{enumerate}
\item For $i\in[N]$, sample $\zeta_{0}^{i}\sim\pi_{0}$ and set $W_{0}^{i}\leftarrow1$.
\item For $n\geq1$:

\begin{enumerate}
\item Create an unpopulated tree $T_{0}$ with root $\nu_{0}$ and leaves
$\left\{ \nu_{i}:i\in[N]\right\} $.
\item For each $i\in[N]$, set 
\[
\mathcal{V}(\nu_{i})\leftarrow\left(1,W_{n-1}^{i}g_{n-1}(\zeta_{n-1}^{i})\right).
\]

\item Call $\mathtt{populate}(\nu_{0})$.
\item Set \texttt{$F\leftarrow\mathtt{choose.forest}(\nu_{0},\tau)$}.
\item For each $i\in[N]$:

\begin{enumerate}
\item Set $W_{n}^{i}\leftarrow\mathcal{V}_{2}(\mathcal{R}(t_{F}(i)))/\mathcal{V}_{1}(\mathcal{R}(t_{F}(i)))$,
\item Set $j\leftarrow\mathtt{sample}(\mathcal{R}(t_{F}(i)))$,
\item Sample $\zeta_{n}^{i}\sim f(\zeta_{n-1}^{j},\cdot)$.\end{enumerate}
\end{enumerate}
\end{enumerate}
\end{algorithm}

\subsection{Partitioning strategies\label{sub:partitioning_strategies} }

The $(\star\star)$ step in Algorithm~\ref{alg:forest_constrained}
remains to be specified. A simple choice would be to choose the partition
$\left\{ \{\chi\}:\chi\in\mathcal{C}(\nu)\right\} $ if it satisfies
the condition in $(\star\star)$ and $\left\{ \mathcal{C}(\nu)\right\} $
otherwise. However, this could lead to more interaction than is necessary.

Before continuing, we note that selecting a partition of child nodes
of $\nu$ is equivalent to selecting a partition $P$ of $\ell(\nu)$
subject to the constraint that the chosen partition is a coarsening
of $P_{0}:=\left\{ \ell(\chi):\chi\in\mathcal{C}(\nu)\right\} $.
Therefore, we simplify the presentation by considering partitions
of $V\subseteq[N]$ instead of partitions of nodes and our goal is
to choose a partition $P\succeq P_{0}$ of $V$ such that $\rho(P,c)\geq\tau$.

If a specific order over coarsenings of $P_{0}$ is defined, one could
seek to find the minimal coarsening $P^{*}$ w.r.t. this order that
satisfies $\rho(P^{*},c)\geq\tau$. For example, one might wish to
find a $P\succeq P_{0}$ subject to $\rho(P,c)\geq\tau$ with the
maximal number of elements, or where the size of the largest element
is minimized, both of which could be translated roughly as $P$ being
as refined as possible. This can always be achieved by enumerating
candidate partitions $P_{1},P_{2},\ldots$ in the given order and
calculating $\rho(P_{i},c)$ for each until some $\rho(P_{i},c)\geq\tau$,
but this can quickly become computationally prohibitive as $\left|\mathcal{C}(\nu)\right|$
grows. Indeed, the number of candidate partitions is the $\left|\mathcal{C}(\nu)\right|$'th
Bell number. This type of integer programming optimization problem
is related to the Partition problem (see, e.g., \citep{mertens2006easiest})
and is likely to be NP-hard in general. We therefore focus on efficient
search strategies for finding a $P\succeq P_{0}$ subject to $\rho(P,c)\geq\tau$
for which we hope that $P$ is not much coarser than necessary.

Both of the strategies we introduce below consider a sequence of successively
coarser partitions $P_{1},P_{2},\ldots$ which satisfy the constraint
that $P_{0}\preceq P_{1}\preceq P_{2}\preceq\cdots$, where $P_{0}$
is as above, and returns $P_{j}$ such that $j=\min\left\{ i:\rho(P_{i},c)\geq\tau\right\} $.
This general procedure has the property that $\rho(P_{i},c)\geq\rho(P_{i-1},c)$
and $\left|P_{i}\right|\leq\left|P_{i-1}\right|-1$ for $i\in\left[\left|\mathcal{C}(\nu)\right|\right]$.
The latter, together with the fact that (from part 1 of Proposition~\ref{prop:ess_properties})
$\left|P\right|=1\implies\rho(P,c)=1\geq\tau$, implies that the total
number of partitions considered is at most $\left|\mathcal{C}(\nu)\right|$.
The specific strategies below are therefore defined by the precise
way in which the sequence $P_{1},P_{2}\ldots$ is chosen.

\subsubsection*{Pairing strategy for structured trees}

This strategy applies when each node in $T_{0}$ has a number of children
that is a power of $2$ and the number of leaves associated with each
child is equal. 
\begin{defn}
(Pairing of a partition) Let $P$ be a partition of $V\subseteq[N]$.
A pairing $P'$ of $P$ is a partition of $V$ where each element
of $P'$ is the union of two elements of $P$.
\end{defn}
\citet[Section~5.4]{whiteley2013role} suggested a ``greedy'' pairing
strategy, which we formalize in the following proposition. 
\begin{prop}
\label{prop:optimal_pairing}Let $P$ be a partition of $V\subseteq[N]$
with $P=\left\{ V_{i}:i\in[2M]\right\} $ for some $M\in[N]$, $M\leq N/2$.
Let $V_{i}$ be ordered such that $0\leq\sum_{j\in V_{1}}c^{j}\leq\cdots\leq\sum_{j\in V_{2M}}c^{j}$
and assume that $\left|V_{i}\right|=\left|V_{j}\right|$ for any $i,j\in[2M]$.
Then a pairing $P'$ of $P$ that maximizes $\rho(P',c)$ is given
by $P'=\left\{ \left\{ V_{1},V_{2M}\right\} ,\left\{ V_{2},V_{2M-1}\right\} ,\ldots,\left\{ V_{M},V_{M+1}\right\} \right\} $.
\end{prop}
In the pairing strategy, then, we define the sequence of partitions
$P_{1},P_{2},\ldots$ by each $P_{i}$ being the optimal pairing of
$P_{i-1}$ provided by Proposition~\ref{prop:optimal_pairing}.

\subsubsection*{Matching strategy}

This strategy does not rely on any particular structure of $T_{0}$
and therefore is applicable more generally than the pairing strategy. 
\begin{prop}
\label{prop:single_match_optimal}For some $K\in[N]$ let $P=\{V_{i}:i\in[K]\}$
be a partition of $V$ and $P_{k,l}:=\{V_{i}:i\in[K]\}\setminus\{V_{k},V_{l}\}\cup\{V_{k}\cup V_{l}\}$
a coarsening of $P$ associated with the indices $k,l\in[K]$. Then
the choice of $k,l\in[K]$ that maximizes $\rho(P_{k,l},c)$ is 
\[
\arg\max_{(k,l)\in[K]^{2}}\frac{\left|V_{k}\right|\left|V_{l}\right|}{\left|V_{k}\right|+\left|V_{l}\right|}\left(\frac{\sum_{j\in V_{k}}c^{j}}{\left|V_{k}\right|}-\frac{\sum_{j\in V_{l}}c^{j}}{\left|V_{l}\right|}\right)^{2}.
\]

\end{prop}
When $[K]$ is large, maximizing this expression by evaluating it
for each $(k,l)\in[K]^{2}$ has a time complexity in $\mathcal{O}(K^{2})$,
which we wish to avoid. Therefore, we resort to finding the $(k,l)\in[K]^{2}$
for which only the squared expression is maximized. This happens when
$k$ and $l$ correspond to the sets of indices whose associated terms
in the squared expression are most different.

The matching strategy therefore defines the successively coarser partitions
$P_{1},P_{2},\ldots$ by letting $S_{i-1}^{\min}=\arg\min_{S\in P_{i-1}}\left|S\right|^{-1}\sum_{j\in S}c^{j}$,
$S_{i-1}^{\max}=\arg\max_{S\in P_{i-1}}\left|S\right|^{-1}\sum_{j\in S}c^{j}$,
and setting 
\[
P_{i}=P_{i-1}\setminus\left\{ S_{i-1}^{\min},S_{i-1}^{\max}\right\} \cup\left\{ S_{i-1}^{\min}\cup S_{i-1}^{\max}\right\} .
\]
An interpretation of this is that the elements of the partition with
whose associated values are most different are successively matched.

\section{Discussion}

\subsection{Numerical illustrations\label{sub:Numerical-illustrations}}

We consider a simplified HMM whose empirical analysis illustrates
the cost of the forest resampling schemes. In particular, we assume
that the HMM equations (\ref{eq:HMM}) satisfy the additional conditional
independence criterion that for any $x\in\mathsf{X}$, $f(x,\cdot)=\pi_{0}(\cdot)$,
and that $g_{n}$ in (\ref{eq:g_n}) is time-homogeneous with $g_{n}(x)=g(x)$.
We further assume that when $X\sim\pi_{0}$, $g(X)$ is a $\ln\mathcal{N}\left(-\frac{\sigma^{2}}{2},\sigma^{2}\right)$
random variable, with mean $1$ and variance $\exp\left(\sigma^{2}\right)-1$.
This model is not intended to be a realistic, challenging application
of SMC. Instead, its greatly simplified structure allows for transparent
analysis and easy replication of results; the time-homogeneous nature
of the model makes it well-suited for assessing the \emph{computational}
cost of resampling for large $n$, and its conditional independence
structure allows us to make some calculations which explicitly show
how the ESS is related to the moments of $Z_{n}^{N}$ and $\pi_{n}^{N}(\varphi)$. 

Writing $\mathbb{E}$ and $\mathbb{V}$ for respectively expectation
and variance under the SMC algorithm, and for some measure $\mu$
and function $\varphi$, ${\rm var}_{\mu}(\varphi):=\int_{\mathsf{X}}\left[\varphi(x)-\mu(\varphi)\right]^{2}\mu(dx)$,
one can verify from $(\dagger)$, (\ref{eq:ESS_defn_front}) and (\ref{eq:particle_approximations})
that $\mathbb{E}\left(Z_{n}^{N}\mid\zeta_{0},\ldots\zeta_{n-1}\right)=\pi_{0}(g)Z_{n-1}^{N}=Z_{n-1}^{N}$
with 
\[
\mathbb{V}\left(\frac{Z_{n}^{N}}{Z_{n-1}^{N}}\mid\zeta_{0},\ldots\zeta_{n-2}\right)=\frac{{\rm var}_{\pi_{0}}\left(g\right)}{N_{n-1}^{\text{eff}}}=\frac{\exp\left(\sigma^{2}\right)-1}{N_{n-1}^{\text{eff}}},
\]
and $\mathbb{E}\left(\pi_{n}^{N}(\varphi)\mid\zeta_{0},\ldots\zeta_{n-1}\right)=\pi_{n}(\varphi)=\pi_{0}(\varphi)$
with 
\[
\mathbb{V}\left(\pi_{n}^{N}(\varphi)\mid\zeta_{0},\ldots\zeta_{n-1}\right)=\frac{{\rm var}_{\pi_{0}}\left(\varphi\right)}{N_{n-1}^{\text{eff}}}.
\]

We define the cost of an $\alpha$SMC resampling step at time $n$
to be the average degree of the vertices in the forest corresponding
to the $\alpha_{n-1}$ transition matrix chosen in $(\star)$ of Algorithm~\ref{alg:aSMC},
which we denote $d_{n}^{N}$. For example, when $\alpha_{n-1}=Id$
the cost is $1$ and when $\alpha_{n-1}=\mathbf{1}_{1/N}$ the cost
is $N$. We ran Algorithm~\ref{alg:aSMC} for $n=200$ iterations
with various values of $\tau$ and $\sigma$ and $N=2^{12}=4096$
particles. One can think of the value of $N$ reported here as being
a large multiple of $4096$ since, conceptually, one could imagine
that the leaves in this experiment represent devices with a large
number of particles. The tree $T_{0}$ used at each iteration always
consisted of three levels with each node except the leaves having
$2^{4}=16$ children, but the leaf/device indices were permuted at
each iteration. 

Figure~\ref{fig:avgcosts_ess} shows the behaviour of $\bar{d}=n^{-1}\sum_{p=1}^{n}d_{p}^{N}$
and $\overline{N^{\text{eff}}}=n^{-1}\sum_{p=1}^{n}N_{p}^{\text{eff}}$
as $\tau$ and $\sigma$ vary using the adaptive resampling particle
filter (ARPF) of \citep{liu1995blind}, and the two proposed strategies
in Section~\ref{sub:partitioning_strategies}, all instances of $\alpha$SMC.
We can see that the ARPF is particularly expensive in terms of average
degree, and has a higher average ESS than the rest. The pairing and
matching strategies perform much better with the latter being less
expensive and having an ESS much closer to the threshold. In all cases,
increases in $\tau$ and $\sigma$ increase the cost of the algorithm,
as one would expect. However, the shape of the curve in Figure~\ref{fig:closeup_dbar}
suggests that increasing $\tau$ beyond around $0.5$ rapidly becomes
expensive. Indeed, the value $\tau=1$ corresponds to an average degree
of $4096$ in this example for any of the methods, which is not shown,
and is almost $10$ times larger than the corresponding cost for $\tau=224/225\approx0.996$,
the rightmost point shown. Figure~\ref{fig:fixed_ESS_increasing_N_match}
shows further that fixing $N\tau=2048$ but increasing $N$ has the
effect of reducing the average degree to close to $2$ and suggests
that optimizing, in terms of computational cost, the choice of $N$
and $\tau$ with a given target ESS $N\tau$ could involve choosing
a large $N$ and a small $\tau$, depending on the relative cost of
increasing $N$ compared to the cost in interactions.

\begin{figure*}
\noindent \centering{}\subfloat[]{\noindent \begin{centering}
\includegraphics[scale=0.55]{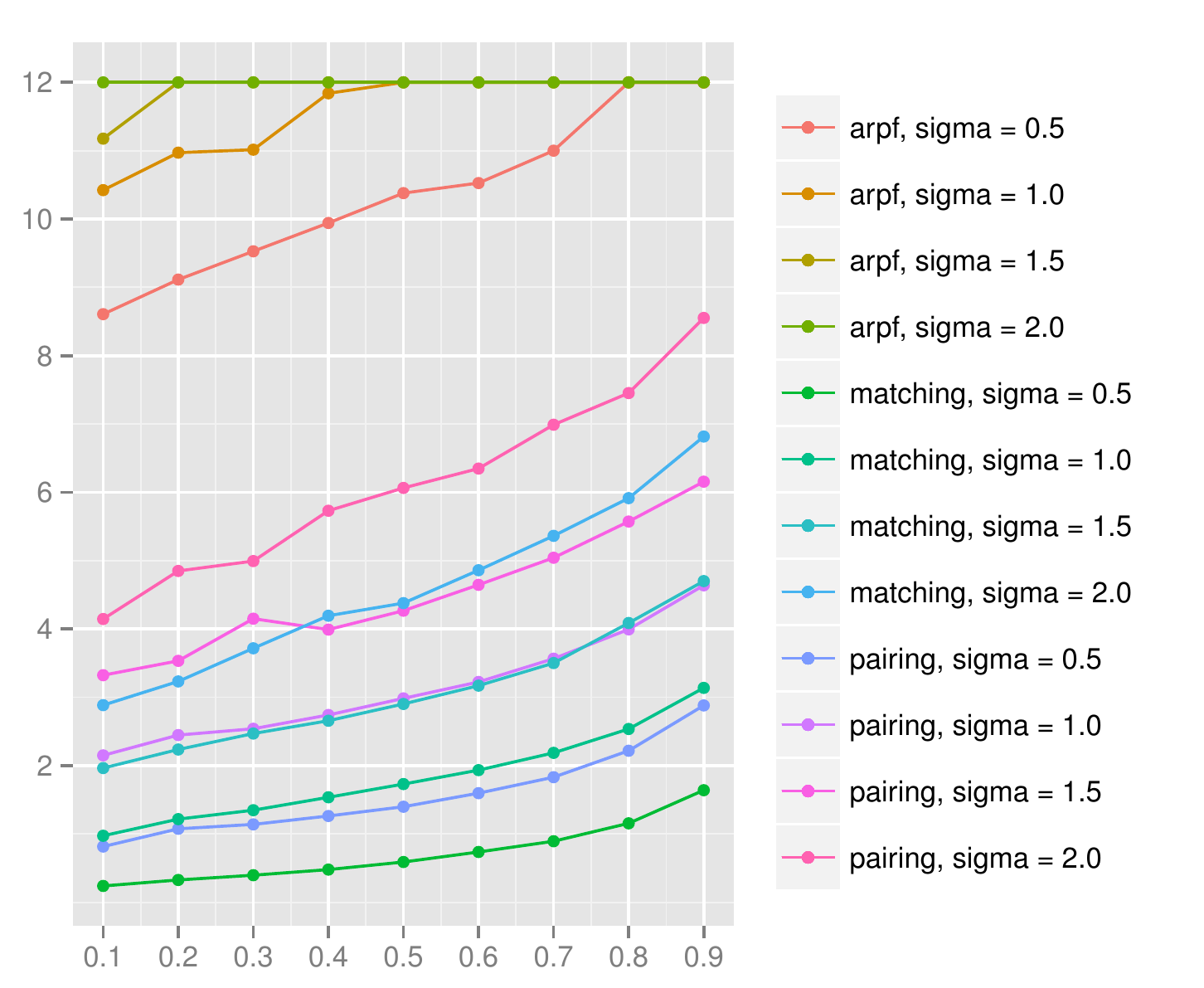}
\par\end{centering}

}\quad{}\subfloat[]{\noindent \begin{centering}
\includegraphics[scale=0.55]{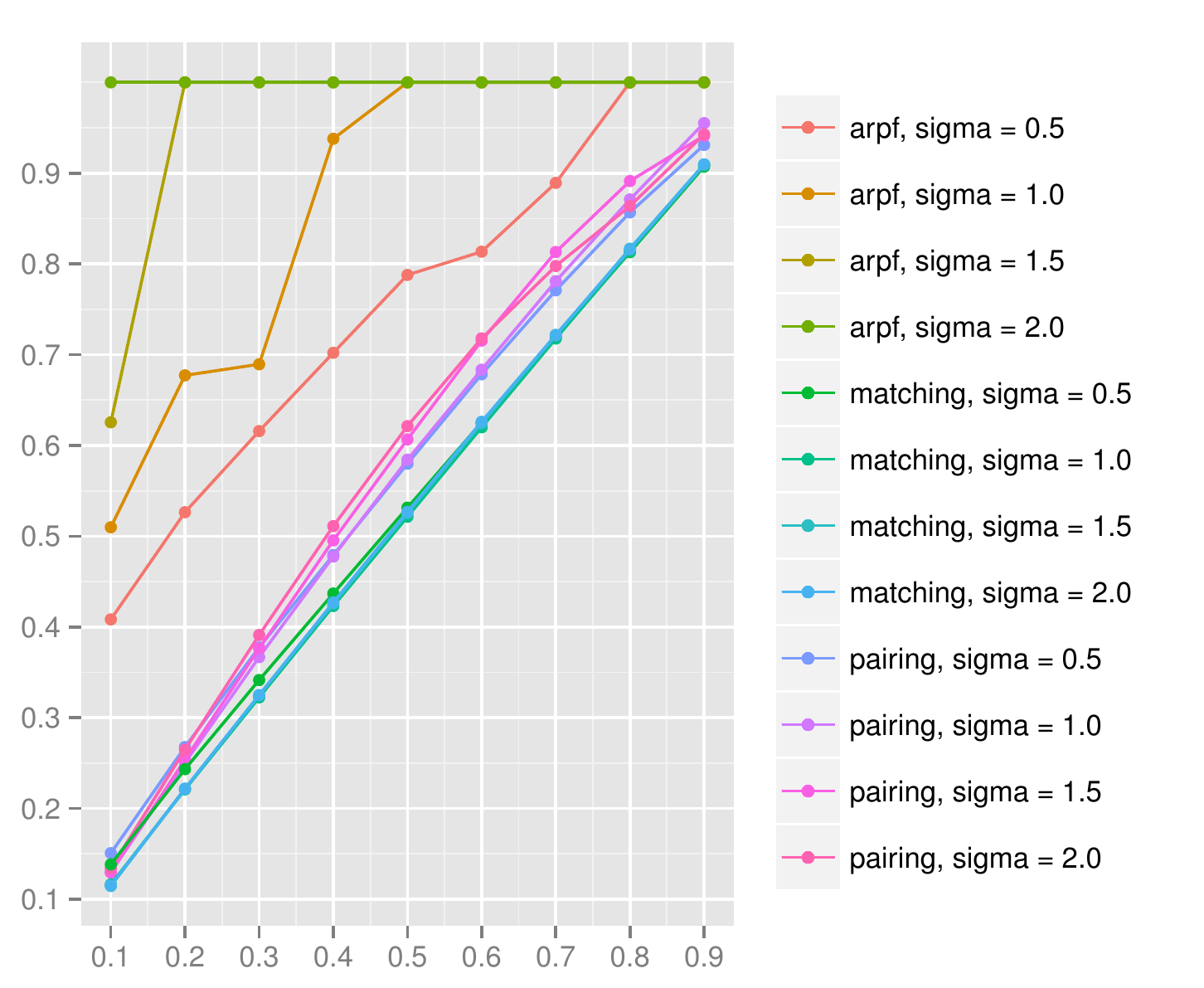}
\par\end{centering}

}\protect\caption{\label{fig:avgcosts_ess}Graphs of (a) $\log_{2}\bar{d}$ and (b)
$\overline{N^{\text{eff}}}/N$ against $\tau$ for the various forest
selection schemes and choices of $\sigma$.}
\end{figure*}

\begin{figure*}
\subfloat[\label{fig:closeup_dbar}Plot of $\bar{d}$ against $\tau$ for $\sigma=1$.]{\noindent \begin{centering}
\includegraphics[scale=0.55]{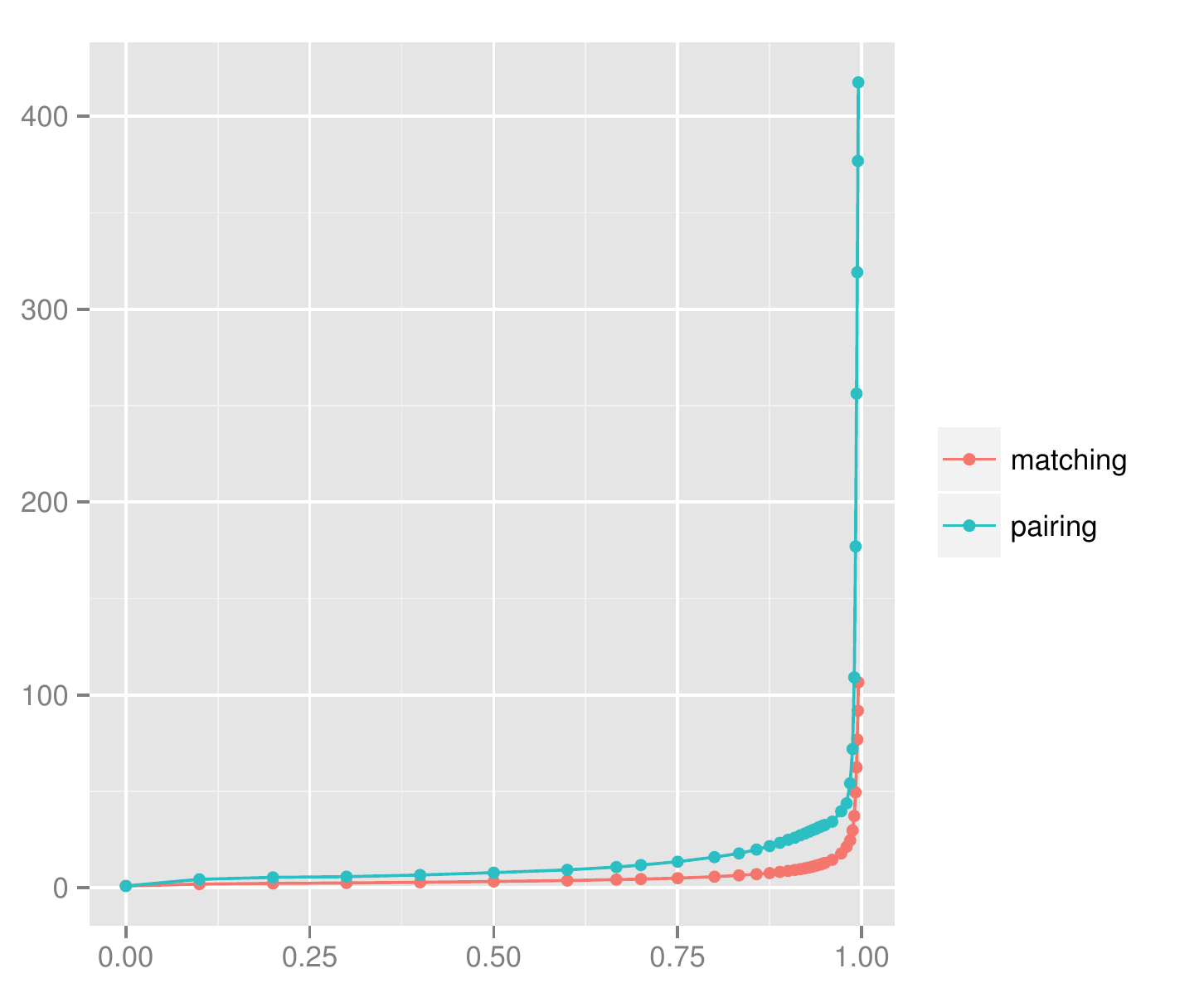}
\par\end{centering}

}\subfloat[\label{fig:fixed_ESS_increasing_N_match}Plot of $\bar{d}$ against
$N$ with $N\tau=2048$ for $\sigma=1$ using the matching strategy.]{\noindent \begin{centering}
\includegraphics[scale=0.55]{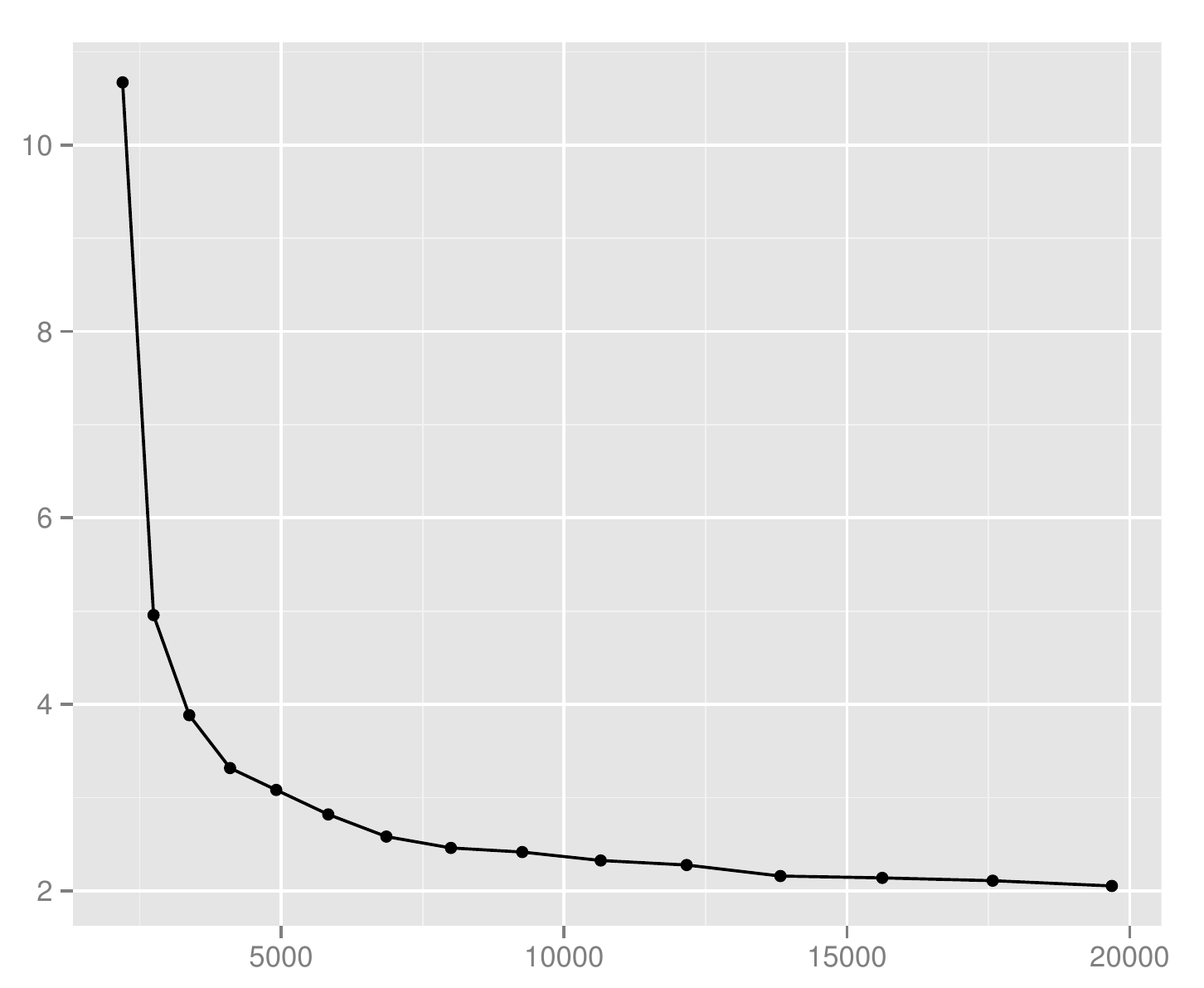}
\par\end{centering}

}

\protect\caption{Dependence of $\bar{d}$ on $\tau$}
\end{figure*}

\subsection{Connection to existing sampling schemes\label{sub:other_resampling}}

Resampling methods other than multinomial can be implemented using
trees as well. In order to make this concrete, we assume that the
tree is ordered, i.e., the children of each node written in sequence
as $\chi_{1},\chi_{2},\ldots$. This imposes only the constraint that
the labelling of children is consistent, and allows the specification
of Algorithm~\ref{alg:select_tree}, which implements Algorithm~\ref{alg:sample_tree}
with a single uniform random variable using the recycling method of
\citep[Section~III.3.7]{devroye1986non}. Proposition~\ref{prop:tree_sampling_functionu}
and the Remark that follows then imply that we can view this algorithm
as a tree-based implementation of the inverse transform method for
sampling from a categorical distribution.

\begin{algorithm}
\protect\caption{Select a value in $\ell(\nu)$ given a $u\in[0,1]$\label{alg:select_tree}}

\texttt{select}$(\nu,u)$
\begin{enumerate}
\item If $\mathcal{C}(\nu)=\emptyset$, return the only element in $\ell(\nu)$.
\item Otherwise, let $\chi_{1},\ldots,\chi_{|\mathcal{C}(\nu)|}$ be the
children of $\nu$ in order.
\item Set 
\[
i\leftarrow\min\left\{ k:\sum_{j=[k]}\mathcal{V}_{2}(\chi_{j})\geq u\sum_{j\in\left[\left|\mathcal{C}(\nu)\right|\right]}\mathcal{V}_{2}(\chi_{j})\right\} .
\]

\item Return 
\[
\mathtt{select}\left(\chi_{i},\frac{u\sum_{j\in\left[\left|\mathcal{C}(\nu)\right|\right]}\mathcal{V}_{2}(\chi_{j})-\sum_{j=[i-1]}\mathcal{V}_{2}(\chi_{j})}{\mathcal{V}_{2}(\chi_{i})}\right).
\]
\end{enumerate}
\end{algorithm}

\begin{prop}
\label{prop:tree_sampling_functionu}Assume that the tree is ordered
such that for each node its children $(\chi_{i})$ have $j\in\ell(\chi_{i}),k\in\ell(\chi_{i+1})\implies j<k$.
Calling Algorithm~\ref{alg:select_tree} with $(\nu,u)$ returns
$\min\left\{ k:\sum_{j\in\ell(\nu)\cap[k]}c^{j}\geq u\sum_{j\in\ell(\nu)}c^{j}\right\} $.\end{prop}
\begin{rem*}
The ordering specified above is w.r.t. the indices of particles and
imposes no real constraint on how the tree is actually constructed,
as long as a specific order is used in step 2 of Algorithm~\ref{alg:select_tree}.
If an alternative ordering is assumed in Proposition~\ref{prop:tree_sampling_functionu}
the resulting returned value will still be deterministic and of the
form given with a slight modification to account for this alternative
ordering.

\end{rem*}
Multinomial resampling corresponds to sampling $N$ i.i.d. uniform
random variables $u^{1},\ldots,u^{N}$ and calling \texttt{select}$(\nu,u^{i})$
for each $i\in[N]$, thereby providing $N$ i.i.d. draws from a categorical
distribution. One can view other resampling methods as making dependent
draws from a categorical distribution by the inverse transform method
by using random variables $u^{1},\ldots,u^{N}$ that are not i.i.d.
but for which the distribution of $u^{K}$, where $K$ is chosen uniformly
at random from $[N]$, is uniform on $[0,1]$ \citep[see, e.g.,][]{Douc2005}.
Therefore, to implement alternative resampling schemes, one again
calls \texttt{select}$(\nu,u^{i})$ for each $i\in[N]$, but with
$u^{1},\ldots,u^{N}$ are distributed in a dependent fashion as in
\citep{Douc2005}. The dependent $u^{i}$ can be interpreted as ``trickling''
down a tree whose leaves represent ancestor indices in a manner reminiscent
of the approach in \citep{crisan2002minimal}, which most closely
resembles the systematic resampling scheme in \citep{Kitagawa1996}.

\subsection{Concluding Remarks}

For ease of presentation, we have chosen to work with a particularly
simple version of $\alpha$SMC, in which new samples are proposed
using the HMM Markov kernel $f$. As noted in \citep{whiteley2013role},
the algorithm is easily generalized to accommodate other proposal
kernels.

This paper, and the methodology of \citep{whiteley2013role} more
generally, naturally complements the contribution of \citep{jun2012entangled}.
In particular, the methods in the latter allow particles to be ``reconstructed''
on a device on the basis of only a small amount of communicated information,
and could be used in tandem with the algorithms here in appropriate
applications.

Both the approaches in Section~\ref{sub:partitioning_strategies}
resemble in some ways greedy strategies for solving the classical
Partition problem. It would be of interest to consider analogues of
more sophisticated solutions to this problem such as those in \citep{karmarkar1982differencing}
and \citep{korf1995approximate}. More generally, it would be of interest
to have quantitative theoretical results enabling the comparison of
particular tree structures and partition selection schemes.

In practice, it may often be the case that devices are homogeneous,
with the network connections between any two devices being of similar
latency and bandwidth. In such situations, one will often create the
structure of the tree at levels above the device and particle layers
in a highly structured way. The use of a randomly generated tree may
be beneficial, as suggested by the Random adaptation rule of \citep[Section~5.4]{whiteley2013role},
which had however no hierarchy. A random permutation of the device
nodes, in an otherwise constant tree was used in Section~\ref{sub:Numerical-illustrations}
for this reason.

Finally, this paper is concerned primarily with $\alpha$ matrices
that are induced by disjoint unions of complete graphs, and hence
have a particular structure. It would be of interest to explore similar
results and methodology for more general $\alpha$ matrices.

\subsection*{Acknowledgements}

We thank Dr. Kari Heine for assistance with the figures. The second
author is supported in part by EPSRC grant EP/K023330/1.

\appendix

\section{Proofs}
\begin{proof}[Proof of Proposition~\ref{prop:ess_properties}]
It is straightforward to show that if $a\in\mathbb{A}_{V}$ then
$\sum_{i}a^{ij}=\sum_{i\in V}a^{ij}=\mathbb{I}(j\in V)$. Therefore,
\begin{eqnarray*}
N^{{\rm eff}}(a,c) & = & \frac{\left(\sum_{i}\sum_{j}a^{ij}c^{j}\right)^{2}}{\sum_{i}\left(\sum_{j}a^{ij}c^{j}\right)^{2}}=\frac{\left(\sum_{i\in V}\sum_{j\in V}a^{ij}c^{j}\right)^{2}}{\sum_{i\in V}\left(\sum_{j\in V}a^{ij}c^{j}\right)^{2}}\\
 & = & \frac{\left(\sum_{j\in V}c^{j}\right)^{2}}{\sum_{i\in V}\left(\sum_{j\in V}a^{ij}c^{j}\right)^{2}}.
\end{eqnarray*}
We define $x_{1}:=\sum_{j\in V}c^{j}$, $x_{2}:=\sum_{j\in V'}c^{j}$,
$y_{1}:=\sum_{i\in V}\left(\sum_{j\in V}a^{ij}c^{j}\right)^{2}$,
$y_{2}:=\sum_{i\in V'}\left(\sum_{j\in V'}\left(a'\right)^{ij}c^{j}\right)^{2}$
and $\tilde{y}_{2}:=\sum_{i\in V'}\left(\sum_{j\in V'}\left(\tilde{a}'\right)^{ij}c^{j}\right)^{2}$.
This allows us to write $N^{{\rm eff}}\left(a,c\right)=x_{1}^{2}/y_{1}$,
$N^{{\rm eff}}\left(a',c\right)=x_{2}^{2}/y_{2}$ and $N^{{\rm eff}}\left(\tilde{a}',c\right)=x_{2}^{2}/\tilde{y}_{2}$,
and since $V\cap V'=\emptyset$, $N^{{\rm eff}}\left(a+a',c\right)=(x_{1}+x_{2})^{2}/(y_{1}+y_{2})$
and $N^{{\rm eff}}\left(a+\tilde{a}',c\right)=(x_{1}+x_{2})^{2}/(y_{1}+\tilde{y}_{2})$.

1. The lower bound holds because
\begin{eqnarray*}
\left(\sum_{i}\sum_{j}a^{ij}c^{j}\right)^{2} & = & \sum_{i,k}\left(\sum_{j}a^{ij}c^{j}\right)\left(\sum_{j}a^{kj}c^{j}\right)\\
 & \geq & \sum_{i}\left(\sum_{j}a^{ij}c^{j}\right)^{2},
\end{eqnarray*}
 so $N^{{\rm eff}}(a,c)\geq1$. The upper bound holds because, using
Jensen's inequality,
\begin{eqnarray*}
\left(\sum_{i}\sum_{j}a^{ij}c^{j}\right)^{2} & = & |V|^{2}\left(\sum_{i\in V}|V|^{-1}\sum_{j\in V}a^{ij}c^{j}\right)^{2}\\
 & \leq & |V|^{2}\sum_{i\in V}|V|^{-1}\left(\sum_{j\in V}a^{ij}c^{j}\right)^{2}\\
 & = & |V|\sum_{i\in V}\left(\sum_{j\in V}a^{ij}c^{j}\right)^{2},
\end{eqnarray*}
 so $N^{{\rm eff}}(a,c)\leq|V|$. The upper bound is attained when
$a^{ij}=|V|^{-1}\mathbb{I}\left(i,j\in V\right)$ since then
\begin{eqnarray*}
\sum_{i\in V}\left(\sum_{j\in V}a^{ij}c^{j}\right)^{2} & = & \sum_{i\in V}\left(\sum_{j\in V}|V|^{-1}c^{j}\right)^{2}\\
 & = & |V|^{-1}\left(\sum_{j\in V}c^{j}\right)^{2}.
\end{eqnarray*}

2. The result follows from
\begin{eqnarray*}
 &  & N^{{\rm eff}}\left(a,c\right)+N^{{\rm eff}}\left(a',c\right)-N^{{\rm eff}}\left(a+a',c\right)\\
 & = & \frac{x_{1}^{2}}{y_{1}}+\frac{x_{2}^{2}}{y_{2}}-\frac{(x_{1}+x_{2})^{2}}{y_{1}+y_{2}}\\
 & = & x_{1}^{2}\left\{ \frac{y_{2}}{y_{1}(y_{1}+y_{2})}\right\} +x_{2}^{2}\left\{ \frac{y_{1}}{y_{2}(y_{1}+y_{2})}\right\} -\frac{2x_{1}x_{2}}{y_{1}+y_{2}}\\
 & = & \frac{y_{1}y_{2}}{y_{1}+y_{2}}\left(\frac{x_{1}}{y_{1}}-\frac{x_{2}}{y_{2}}\right)^{2}\geq0,
\end{eqnarray*}
with equality only when $\frac{x_{1}}{y_{1}}=\frac{x_{2}}{y_{2}}$,
corresponding to $\frac{\sum_{j\in V}c^{j}}{N^{{\rm eff}}(a,c)}=\frac{\sum_{j\in V'}c^{j}}{N^{{\rm eff}}(a',c)}$.

3. Since $N^{{\rm eff}}(a',c)\leq N^{{\rm eff}}(\tilde{a}',c)\implies y_{2}\geq\tilde{y}_{2}$,
\begin{eqnarray*}
 &  & N^{{\rm eff}}(a+\tilde{a}',c)-N^{{\rm eff}}(a+a',c)\\
 &  & =\frac{\left(x_{1}+x_{2}\right)^{2}}{y_{1}+\tilde{y}_{2}}-\frac{\left(x_{1}+x_{2}\right)^{2}}{y_{1}+y_{2}}\geq0,
\end{eqnarray*}
with equality only when $y_{2}=\tilde{y}_{2}$, corresponding to $N^{{\rm eff}}(a',c)\leq N^{{\rm eff}}(\tilde{a}',c)$.

4. We have
\begin{eqnarray*}
 &  & N^{{\rm eff}}\left(a_{1}+a_{2},c\right)\\
 &  & =\frac{(x_{1}+x_{2})^{2}}{y_{1}+y_{2}}\\
 &  & =\frac{x_{1}^{2}}{y_{1}}\cdot\frac{y_{1}}{y_{1}+y_{2}}+\frac{2x_{1}x_{2}}{y_{1}+y_{2}}+\frac{x_{2}^{2}}{y_{2}}\cdot\frac{y_{2}}{y_{1}+y_{2}}\\
 &  & \geq\min_{i}\left\{ \frac{x_{i}^{2}}{y_{i}}\right\} =\min\left\{ N^{{\rm eff}}\left(a,c\right),N^{{\rm eff}}\left(a',c\right)\right\} .
\end{eqnarray*}

\end{proof}

\begin{proof}[Proof of Proposition~\ref{prop:partial_order}]
We prove the result for $\mathbb{G}$ since the result for $\mathbb{A}_{\mathbb{G}}$
then follows. Let $V\subseteq[N]$ and consider $G_{1}=(V,E_{1})$,
$G_{2}=(V,E_{2})$ and $G_{3}=(V,E_{3})$. It suffices to check that
$\preceq$ is reflexive ($G_{1}\preceq G_{1})$, antisymmetric ($G_{1}\preceq G_{2}$
and $G_{2}\preceq G_{1}$ implies $G_{1}=G_{2}$) and transitive ($G_{1}\preceq G_{2}$
and $G_{2}\preceq G_{3}$ implies $G_{1}\preceq G_{3}$). Since $E_{1}\subseteq E_{1}$,
it follows that $G_{1}\preceq G_{1}$. When $G_{1}\preceq G_{2}$
and $G_{2}\preceq G_{1}$, this implies $E_{1}\subseteq E_{2}$ and
$E_{2}\subseteq E_{1}$ and it follows that $E_{1}=E_{2}$ and so
$G_{1}=G_{2}$. Finally, $G_{1}\preceq G_{2}$ and $G_{2}\preceq G_{3}$
implies that $E_{1}\subseteq E_{2}$ and $E_{2}\subseteq E_{3}$ and
so $E_{1}\subseteq E_{3}$ and therefore $G_{1}\preceq G_{3}$.
\end{proof}

\begin{proof}[Proof of Proposition~\ref{prop:order_preserving}]
Since $a\preceq\tilde{a}$ we have that $a,\tilde{a}\in\mathbb{A}_{\mathbb{G}}\cap\mathbb{A}_{V}$
for some $V\subseteq[N]$. Therefore, for some $K,\tilde{K}\in[N]$
we can write $a=\sum_{k\in[K]}\phi(\kappa(V_{k}))$ and $\tilde{a}=\sum_{\tilde{k}\in[\tilde{K}]}\phi(\kappa(\tilde{V}_{\tilde{k}}))$
where each $V_{k}$ and $\tilde{V}_{\tilde{k}}$ are subsets of $V$.
Since $a\preceq\tilde{a}$, for each $k\in[K]$ there exists $\tilde{k}\in[\tilde{K}]$
such that $V_{k}\subseteq\tilde{V}_{\tilde{k}}$. We now define a
sequence, with $a_{0}=a$, and for $i\in[\tilde{K}]$
\[
a_{i}:=a_{i-1}+\phi(\kappa(\tilde{V}_{i}))-\sum_{k\in[K],V_{k}\cap\tilde{V}_{i}\neq\emptyset}\phi(\kappa(V_{k})),
\]
 and note that $a_{\tilde{K}}=\tilde{a}$. Now for each $i\in[\tilde{K}]$,
$a_{i}\in\mathbb{A}_{\mathbb{G}}\cap\mathbb{A}_{V}$ and letting $\check{V}_{i}:=\bigcup_{j=1}^{i}\tilde{V}_{j}$
and 
\[
b_{i}:=\sum_{k\in[K],V_{k}\cap\check{V}_{i}=\emptyset}\phi(\kappa(V_{k}))+\sum_{j=1}^{i-1}\phi(\kappa(\tilde{V}_{j})),
\]
we can write $a_{i-1}=b_{i}+\sum_{k\in[K],V_{k}\cap\tilde{V}_{i}\neq\emptyset}\phi(\kappa(V_{k}))$
and $a_{i}=b_{i}+\phi(\kappa(\tilde{V}_{i}))$. From the first part
of Proposition~\ref{prop:ess_properties} we have that $N^{{\rm eff}}\left(\sum_{k\in[K],V_{k}\cap\tilde{V}_{i}\neq\emptyset}\phi(\kappa(V_{k})),c\right)\leq\left|\tilde{V}_{i}\right|=\phi(\kappa(\tilde{V}_{i}))$
and so by the monotonicity property in Proposition~\ref{prop:ess_properties}
we have $N^{{\rm eff}}(a_{i-1},c)\leq N^{{\rm eff}}(a_{i},c)$ for
each $i\in[\tilde{K}]$. It follows that $N^{{\rm eff}}(a,c)\leq N^{{\rm eff}}(\tilde{a},c)$.
\end{proof}

\begin{proof}[Proof of Proposition~\ref{prop:alpha_lower_bound}]
The first inequality follows from Proposition~\ref{prop:order_preserving}
since $a\preceq\tilde{a}$. For the second inequality, assume that
$\min_{k}\left|V_{k}\right|^{-1}N^{{\rm eff}}(a_{k},c)\geq D$. This
implies that for any $k\in[K]$,
\[
\sum_{i\in V_{k}}\left(\sum_{j\in V_{k}}a_{k}^{ij}c^{j}\right)^{2}\leq\frac{1}{D|V_{k}|}\left(\sum_{j\in V_{k}}c^{j}\right)^{2}.
\]
Therefore
\begin{eqnarray*}
N^{{\rm eff}}\left(a,c\right) & = & \frac{\left(\sum_{k\in[K]}\sum_{j\in V_{k}}c^{j}\right)^{2}}{\sum_{k\in[K]}\sum_{i\in V_{k}}\left(\sum_{j\in V_{k}}a_{k}^{ij}c^{j}\right)^{2}}\\
 & \geq & D\frac{\left(\sum_{k\in[K]}\sum_{j\in V_{k}}c^{j}\right)^{2}}{\sum_{k\in[K]}|V_{k}|^{-1}\left(\sum_{j\in V_{k}}c^{j}\right)^{2}}\\
 & = & DN^{{\rm eff}}\left(\tilde{a},c\right).
\end{eqnarray*}

\end{proof}

\begin{proof}[Proof of Proposition~\ref{prop:recursiveboundinduction}]
The proof is by induction. Note that $\tau\in[0,1]$. If $|V|=1$
then $|P|=1$ and $N_{c}^{{\rm eff}}(P)=1$, so the claim is true.
Now assume that the claim holds true for all $V$ with $|V|\in[s-1]$,
$s\in\mathbb{N}$, and consider the case where $|V|=s$. First, note
that if $P=\{V\}$ then $N^{{\rm eff}}(P,c)=|V|$ by Proposition~\ref{prop:ess_properties}
and so if $|P|=1$ the claim is true. It remains to check that if
$|P|>1$ then $a=\sum_{k\in[K]}a_{k}$ satisfies the claim, where
$a_{k}={\tt choose.a}(V_{k},\tau/\rho(P,c))$. By the induction hypothesis,
for each $k\in[K]$, $N^{{\rm eff}}(a_{k},c)/\left|V_{k}\right|\geq\tau/\rho(P,c)$
since $|V_{k}|\in[s-1]$ and $\tau/\rho(P,c)\in[0,1]$. Then by Proposition~\ref{prop:alpha_lower_bound},
with $\tilde{a}=\sum_{k\in[K]}\phi(\kappa(V_{k}))$,
\begin{eqnarray*}
N^{{\rm eff}}\left(a,c\right) & \geq & \min_{k}\left\{ \frac{N^{{\rm eff}}(a_{k},c)}{|V_{k}|}\right\} N^{{\rm eff}}\left(\tilde{a},c\right)\\
 & \geq & \frac{\tau}{\rho(P,c)}\rho(P,c)\left|V\right|=\tau\left|V\right|,
\end{eqnarray*}
and we conclude.
\end{proof}

\begin{proof}[Proof of Proposition~\ref{prop:tree_sample_correctness}]
Let $s_{\nu}(j)$ denote the probability that Algorithm~\ref{alg:sample_tree}
returns $j\in\ell(\nu)$. Given $j\in\ell(\nu)$, let $\nu^{1}$ be
the parent of $\nu_{j}$, $\nu^{2}$ be the parent of $\nu^{1}$,
etc., until $\nu^{m}=\nu$ is the parent of $\nu^{m-1}$. Then 
\begin{eqnarray*}
s_{\nu}(j) & = & \frac{c^{j}}{\sum_{k\in\ell(\nu^{1})}c^{k}}\prod_{i\in[m-1]}\frac{\sum_{k\in\ell(\nu^{i})}c^{k}}{\sum_{k\in\ell(\nu^{i+1})}c^{k}},\\
 & = & \frac{c^{j}}{\sum_{k\in\ell(\nu)}c^{k}}=p_{\mathcal{S}(\nu)}(j).
\end{eqnarray*}

\end{proof}

\begin{proof}[Proof of Proposition~\ref{prop:optimal_pairing}]
We define $x_{i}:=\sum_{j\in V_{i}}c^{j}$ for $i\in[2M]$ and it
suffices to show that $P'$ minimizes the denominator of $\rho(\cdot,c)$,
\[
r(P',c):=\left(2\left|V_{1}\right|\right)^{-1}\sum_{S\in P'}\left(\sum_{i\in S}x_{i}\right)^{2},
\]
since each element of any pairing of $P$ is of size $2\left|V_{1}\right|$.
We first prove that $V_{1}\cup V_{2M}$ is a member of at least one
pairing of $P$ that minimizes $r(\cdot,c)$. Indeed, assume that
a pairing $\check{P}$ that minimizes $r(\cdot,c)$ is given. We will
show that a pairing $\check{P}'$ containing $V_{1}\cup V_{2M}$ exists
for which $r(\check{P}',c)\leq r(\check{P},c)$. Let $V_{1}\cup V_{j}$
and $V_{k}\cup V_{2M}$ be elements of $\check{P}$. We define $\check{P}'=\check{P}\setminus\{V_{1}\cup V_{j},V_{k}\cup V_{2M}\}\cup\{V_{1}\cup V_{2M},V_{k}\cup V_{j}\}$.
Then
\begin{eqnarray*}
 &  & 2\left|V_{1}\right|\left(r(\check{P},c)-r(\check{P}',c)\right)\\
 & = & \sum_{S\in\check{P}}\left(\sum_{i\in S}x_{i}\right)^{2}-\sum_{S\in\check{P}'}\left(\sum_{i\in S}x_{i}\right)^{2}\\
 & = & \left(x_{1}+x_{j}\right)^{2}+\left(x_{k}+x_{2M}\right)^{2}\\
 &  & -\left(x_{1}+x_{2M}\right)^{2}-\left(x_{k}+x_{j}\right)^{2}\\
 & = & 2x_{1}x_{j}+2x_{k}x_{2M}-2x_{1}x_{2M}-2x_{k}x_{j}\\
 & = & 2x_{k}(x_{2M}-x_{j})-2x_{1}(x_{2M}-x_{j})\\
 & = & 2(x_{k}-x_{1})(x_{2M}-x_{j})\geq0,
\end{eqnarray*}
since $x_{1}$ and $x_{2M}$ are the minimal and maximal values of
$\{x_{i}:i\in[2M]\}$, respectively.

Now, let $\check{P}'$ be a pairing of $P$ and $S=V_{1}\cup V_{2M}\in\check{P}'$.
Then $r(\check{P}',c)=r(\{S\},c)+r(\check{P}'\setminus\{S\},c)$ and
so it follows that if $\check{P}'$ minimizes $r(\cdot,c)$ then $\check{P}'\setminus\{S\}$
is a pairing of $P\setminus\left\{ V_{1},V_{2M}\right\} $ that minimizes
$r(\cdot,c)$. It then follows that at least one pairing $\check{P}'$
of $P$ that minimizes $r(\cdot,c)$ is the union of $V_{1}\cup V_{2M}$
and a pairing of $P\setminus\left\{ V_{2M},V_{1}\right\} $ that minimizes
$r(\cdot,c)$. But then the argument above shows that $V_{2}\cup V_{2M-1}$
is a valid element of such a minimizing pairing. Continuing, we obtain
that $P'$ is a pairing of $P$ that minimizes $r(\cdot,c)$ and we
conclude.
\end{proof}

\begin{proof}[Proof of Proposition~\ref{prop:single_match_optimal}]
From (\ref{eq:fp_first}) we can write $\rho(P{}_{k,l},c)=\rho(P,c)\frac{r(P,c)}{r(P_{k,l},c)}$
where $r(P,c):=\sum_{S\in P}|S|\left(|S|^{-1}\sum_{j\in S}c^{j}\right)^{2}$.
It suffices therefore to find $k,l\in[K]$ minimizing $r(P_{k,l},c)$.
Letting $m_{i}=\left|V_{i}\right|$ and $x_{i}=\sum_{j\in V_{i}}c^{j}$,
we can write 
\begin{eqnarray*}
r(P_{k,l},c) & = & r(P,c)-\frac{x_{k}^{2}}{m_{k}}-\frac{x_{l}^{2}}{m_{l}}+\frac{\left(x_{k}+x_{l}\right)^{2}}{m_{k}+m_{l}}\\
 & = & r(P,c)-\frac{m_{k}m_{l}}{m_{k}+m_{l}}\left(\frac{x_{k}}{m_{k}}-\frac{x_{l}}{m_{l}}\right)^{2},
\end{eqnarray*}
the equality following along the same lines as the proof of the second
part of Proposition~\ref{prop:ess_properties}, and we conclude.
\end{proof}

\begin{proof}[Proof of Proposition~\ref{prop:tree_sampling_functionu}]
Let $\chi_{\nu}:=\left(\chi_{\nu,1},\ldots,\chi_{\nu,K}\right)$
be the ordered children of $\nu$, where $|\mathcal{C}(\nu)|=K$.
To alleviate notation, we define $c(\chi_{\nu,k}):=\sum_{j\in\ell(\chi_{\nu,k})}c^{j}$,
$k_{\nu}(u):=\min\left\{ k:\sum_{j\in[k]}c(\chi_{\nu,j})\geq u\sum_{j\in\ell(\nu)}c^{j}\right\} $
and $s_{\nu}(k):=\sum_{j\in[k]}c(\chi_{\nu,j})$. Algorithm~\ref{alg:select_tree}
with input $(\nu,u)$ returns $f_{\nu}(u)$, where
\[
f_{\nu}(u):=\begin{cases}
\min\ell(\nu) & \left|\ell(\nu)\right|=1,\\
f_{\chi_{\nu,k_{\nu}(u)}}\left(\frac{u\sum_{j\in\ell(\nu)}c^{j}-s_{\nu}(k_{\nu}(u)-1)}{c(\chi_{\nu,k_{\nu}(u)})}\right) & \text{otherwise.}
\end{cases}
\]
Now we prove by induction that $f_{\nu}(u)=\min\left\{ k:\sum_{j\in\ell(\nu)\cap[k]}c^{j}\geq u\sum_{j\in\ell(\nu)}c^{j}\right\} $.
If $|\ell(\nu)|=1$, the claim is trivially true. Now assume the claim
is true for $\{\nu:|\ell(\nu)|\in[p-1]\}$ and consider $\nu$ with
$|\ell(\nu)|=p>1$. We have
\[
f_{\nu}(u)=f_{\chi_{\nu,k_{\nu}(u)}}\left(\frac{u\sum_{j\in\ell(\nu)}c^{j}-s_{\nu}(k_{\nu}(u)-1)}{c(\chi_{\nu,k_{\nu}(u)})}\right)
\]
and we can apply the induction hypothesis since $|\ell(\chi_{\nu,k_{\nu}(u)})|\in[p-1]$.
Therefore, letting $y:=u\sum_{j\in\ell(\nu)}c^{j}-s_{\nu}(k_{\nu}(u)-1)$,
and $x:=\chi_{\nu,k_{\nu}(u)}$we can write $f_{\nu}(u)$ as
\begin{eqnarray*}
 &  & \min\left\{ k:\sum_{j\in\ell(x)\cap[k]}c^{j}\geq\frac{y}{c(x)}\sum_{j\in\ell(x)}c^{j}\right\} \\
 & = & \min\left\{ k:\sum_{j\in\ell(x)\cap[k]}c^{j}\geq u\sum_{j\in\ell(\nu)}c^{j}-s_{\nu}(k_{\nu}(u)-1)\right\} \\
 & = & \min\left\{ k:\sum_{j\in\ell(\nu)\cap[k]}c^{j}\geq u\sum_{j\in\ell(\nu)}c^{j}\right\} ,
\end{eqnarray*}
and we conclude.
\end{proof}
\bibliographystyle{unsrtnat}
\bibliography{sadm}

\begin{thebibliography}{26}
\providecommand{\natexlab}[1]{#1}
\providecommand{\url}[1]{\texttt{#1}}
\expandafter\ifx\csname urlstyle\endcsname\relax
  \providecommand{\doi}[1]{doi: #1}\else
  \providecommand{\doi}{doi: \begingroup \urlstyle{rm}\Url}\fi

\bibitem[Doucet et~al.(2000)Doucet, Godsill, and Andrieu]{doucet2000sequential}
A.~Doucet, S.~Godsill, and C.~Andrieu.
\newblock On sequential {Monte Carlo} sampling methods for {Bayesian}
  filtering.
\newblock \emph{Stat. Comput.}, 10\penalty0 (3):\penalty0 197--208, 2000.

\bibitem[Doucet and Johansen(2008)]{Doucet2008}
A.~Doucet and A.~M. Johansen.
\newblock A tutorial on particle filtering and smoothing: Fifteen years later.
\newblock In D.~Crisan and B.~Rozovsky, editors, \emph{Handbook of Nonlinear
  Filtering}. Oxford University Press, 2008.

\bibitem[Chopin(2002)]{smc:meth:C02}
N.~Chopin.
\newblock A sequential particle filter method for static models.
\newblock \emph{Biometrika}, 89\penalty0 (3):\penalty0 539--552, 2002.

\bibitem[Del~Moral et~al.(2006)Del~Moral, Doucet, and Jasra]{DelMoral2006}
P.~Del~Moral, A.~Doucet, and A.~Jasra.
\newblock Sequential {M}onte {C}arlo samplers.
\newblock \emph{J. R. Stat. Soc. Ser. B Stat. Methodol.}, 68\penalty0
  (3):\penalty0 411--436, 2006.

\bibitem[Chopin et~al.(2013)Chopin, Jacob, and
  Papaspiliopoulos]{chopin2013smc2}
N.~Chopin, P.~E. Jacob, and O.~Papaspiliopoulos.
\newblock {SMC}$^2$: an efficient algorithm for sequential analysis of state
  space models.
\newblock \emph{J. R. Stat. Soc. Ser. B Stat. Methodol.}, 75\penalty0
  (3):\penalty0 397--426, 2013.

\bibitem[Suchard and Rambaut(2009)]{Suchard2009}
M.~A. Suchard and A.~Rambaut.
\newblock Many-core algorithms for statistical phylogenetics.
\newblock \emph{Bioinformatics}, 25\penalty0 (11):\penalty0 1370--1376, 2009.

\bibitem[Lee et~al.(2010)Lee, Yau, Giles, Doucet, and Holmes]{lee2010utility}
A.~Lee, C.~Yau, M.~B. Giles, A.~Doucet, and C.~C. Holmes.
\newblock On the utility of graphics cards to perform massively parallel
  simulation of advanced {Monte Carlo} methods.
\newblock \emph{J. Comput. Graph. Statist.}, 19\penalty0 (4):\penalty0
  769--789, 2010.

\bibitem[Boli\'{c} et~al.(2005)Boli\'{c}, Djuri\'{c}, and
  Hong]{bolic2005resampling}
M.~Boli\'{c}, P.~M. Djuri\'{c}, and S.~Hong.
\newblock Resampling algorithms and architectures for distributed particle
  filters.
\newblock \emph{IEEE Trans. Signal Process.}, 53\penalty0 (7):\penalty0
  2442--2450, 2005.

\bibitem[Jun et~al.(2012)Jun, Wang, and
  Bouchard-C{\^o}t{\'e}]{jun2012entangled}
S.-H. Jun, L.~Wang, and A.~Bouchard-C{\^o}t{\'e}.
\newblock Entangled monte carlo.
\newblock In \emph{Advances in Neural Information Processing Systems}, pages
  2726--2734, 2012.

\bibitem[Verg{\'e} et~al.()Verg{\'e}, Dubarry, Del~Moral, and
  Moulines]{verge2013parallel}
C.~Verg{\'e}, C.~Dubarry, P.~Del~Moral, and E.~Moulines.
\newblock On parallel implementation of sequential {M}onte {C}arlo methods: the
  island particle model.
\newblock \emph{Stat. and Comput.}
\newblock To appear.

\bibitem[Whiteley et~al.(2013)Whiteley, Lee, and Heine]{whiteley2013role}
N.~Whiteley, A.~Lee, and K.~Heine.
\newblock On the role of interaction in sequential {M}onte {C}arlo algorithms.
\newblock \emph{arXiv preprint 1309.2918}, 2013.

\bibitem[Liu and Chen(1995)]{liu1995blind}
J.~S. Liu and R.~Chen.
\newblock Blind deconvolution via sequential imputations.
\newblock \emph{J. Amer. Statist. Assoc.}, 90\penalty0 (430):\penalty0
  567--576, 1995.

\bibitem[{Del Moral} and Guionnet(2001)]{smc:the:DMG01}
P.~{Del Moral} and A.~Guionnet.
\newblock On the stability of interacting processes with applications to
  filtering and genetic algorithms.
\newblock \emph{Ann. Inst. Henri Poincar\'{e} Probab. Stat.}, 37\penalty0
  (2):\penalty0 155--194, 2001.

\bibitem[C\'{e}rou et~al.(2011)C\'{e}rou, {Del Moral}, and
  Guyader]{smc:the:CdMG11}
F.~C\'{e}rou, P.~{Del Moral}, and A.~Guyader.
\newblock A nonasymptotic variance theorem for unnormalized {Feynman} {Kac}
  particle models.
\newblock \emph{Ann. Inst. Henri Poincar\'{e} Probab. Stat.}, 47\penalty0
  (3):\penalty0 629--649, 2011.

\bibitem[Whiteley and Lee(2014)]{whiteley2014twisted}
N.~Whiteley and A.~Lee.
\newblock Twisted particle filters.
\newblock \emph{Ann. Statist.}, 42\penalty0 (1):\penalty0 115--141, 2014.

\bibitem[Whiteley(2013)]{whiteley2013stability}
N.~Whiteley.
\newblock Stability properties of some particle filters.
\newblock \emph{Ann. Appl. Probab.}, 23\penalty0 (6):\penalty0 2500--2537,
  2013.

\bibitem[Knuth(1997)]{knuth2005art}
D.~E. Knuth.
\newblock \emph{The Art of Computer Programming}, volume~1.
\newblock Addison-Wes, 3rd edition, 1997.

\bibitem[Hillis and Steele~Jr(1986)]{hillis1986data}
W.~D. Hillis and G.~L. Steele~Jr.
\newblock Data parallel algorithms.
\newblock \emph{Communications of the ACM}, 29\penalty0 (12):\penalty0
  1170--1183, 1986.

\bibitem[Isard et~al.(2007)Isard, Budiu, Yu, Birrell, and
  Fetterly]{isard2007dryad}
M.~Isard, M.~Budiu, Y.~Yu, A.~Birrell, and D.~Fetterly.
\newblock Dryad: distributed data-parallel programs from sequential building
  blocks.
\newblock \emph{ACM SIGOPS Operating Systems Review}, 41\penalty0 (3):\penalty0
  59--72, 2007.

\bibitem[Mertens(2006)]{mertens2006easiest}
S.~Mertens.
\newblock The easiest hard problem: number partitioning.
\newblock In A.~Percus, G.~Istrate, and C.~Moore, editors, \emph{Computational
  Complexity and Statistical Physics}, pages 125--139. Oxford University Press,
  2006.

\bibitem[Devroye(1986)]{devroye1986non}
L.~Devroye.
\newblock \emph{Non-uniform random variate generation}.
\newblock Springer Verlag, 1986.

\bibitem[Douc et~al.(2005)Douc, Capp{\'e}, and Moulines]{Douc2005}
R.~Douc, O.~Capp{\'e}, and E.~Moulines.
\newblock Comparison of resampling schemes for particle filtering.
\newblock In \emph{Proceedings of the 4th International Symposium on Image and
  Signal Processing and Analysis}, pages 64--69, 2005.

\bibitem[Crisan and Lyons(2002)]{crisan2002minimal}
D.~Crisan and T.~Lyons.
\newblock Minimal entropy approximations and optimal algorithms for the
  filtering problem.
\newblock \emph{Monte Carlo methods and applications}, 8\penalty0 (4):\penalty0
  343--356, 2002.

\bibitem[Kitagawa(1996)]{Kitagawa1996}
G.~Kitagawa.
\newblock Monte {C}arlo filter and smoother for non-{G}aussian nonlinear state
  space models.
\newblock \emph{J. Comput. Graph. Statist.}, 5\penalty0 (1):\penalty0 1--25,
  1996.

\bibitem[Karmarkar and Karp(1982)]{karmarkar1982differencing}
N.~Karmarkar and R.~M. Karp.
\newblock The differencing method of set partitioning.
\newblock Technical report, University of California, Berkeley, 1982.

\bibitem[Korf(1995)]{korf1995approximate}
R.~E. Korf.
\newblock From approximate to optimal solutions: A case study of number
  partitioning.
\newblock In \emph{Proceedings of the 14th international joint conference on
  Artificial intelligence}, pages 266--272, 1995.

\end{thebibliography}

\end{document}